\documentclass[lettersize,journal]{IEEEtran}
\makeatletter

\let\MYcaption\@makecaption
\makeatother
\usepackage[font=footnotesize]{subcaption}
\makeatletter
\let\@makecaption\MYcaption
\makeatother
\usepackage{array}
\usepackage{textcomp}
\usepackage{url}
\usepackage{verbatim}
\usepackage{graphicx}
\def\BibTeX{{\rm B\kern-.05em{\sc i\kern-.025em b}\kern-.08em
    T\kern-.1667em\lower.7ex\hbox{E}\kern-.125emX}}
\usepackage{balance}
\usepackage[english]{babel}
\usepackage{graphicx}
\usepackage[normalem]{ulem}
\usepackage{amsmath}
\usepackage{amsthm}
\usepackage{amssymb}
\usepackage{amsfonts}
\usepackage{enumerate}
\usepackage{microtype}
\usepackage[utf8]{inputenc}
\usepackage{blindtext}
\usepackage{cite}
\usepackage{tikz,tkz-euclide}
\usetikzlibrary{backgrounds,automata}
\usetikzlibrary{arrows,snakes,shapes,positioning,arrows,decorations.markings,calc}
\usepackage{color}
\definecolor{green}{rgb}{0,0.4,0.05}
\definecolor{red}{rgb}{0.8,0,0}
\usepackage{float}
\usepackage{cleveref}
\usepackage{slashbox}
\usepackage{pgfplots}
\usepackage{authblk}
\usepackage{amsbsy}
\usepackage{fixmath}
\usepackage{amssymb}
\usetikzlibrary{spy}
\usepackage{lipsum}
\usepackage{mathtools}
\usepackage{cuted}
\usepackage{bibentry}
\usepackage{psfrag}
\usepackage{algorithm}
\usepackage{algpseudocode}
\usepackage{scalefnt}
\usepackage{graphicx}
\usepackage{mathtools}
\usepackage{xcolor}
\usepackage[nolist,nohyperlinks]{acronym}
\DeclarePairedDelimiter{\ceil}{\lceil}{\rceil}
\DeclarePairedDelimiter{\floor}{\lfloor}{\rfloor}
\newsavebox{\measurebox}
\usepackage{pgfplots}
\pgfplotsset{compat=newest}
\usetikzlibrary{plotmarks}
\usetikzlibrary{positioning}
\usepackage{grffile}
\usepackage[european, straightvoltages, RPvoltages, americanresistor, americaninductors]{circuitikz}
\usepackage[compact]{titlesec}

\tikzset{every picture/.style={line width=0.2mm}}
\usetikzlibrary{calc} 	
\ctikzset{bipoles/thickness=1.2}

\newcommand{\cvec}[1]{{\mathbf #1}}

\newcommand{\bm}[1]{\mathbf #1}

\newcommand{\eye}{\mathbf{I}}

\newcommand{\bs}{\boldsymbol}
\newcommand{\op}{\operatorname}

\theoremstyle{definition}
\newtheorem{theorem}{Theorem}

\newtheorem{lemma}{Lemma}
\newlength\fheight 
\newlength\fwidth 
\algdef{SE}[SUBALG]{Indent}{EndIndent}{}{\algorithmicend\ }%
\algtext*{Indent}
\algtext*{EndIndent}

\AtBeginDocument{
\addtolength{\abovedisplayskip}{-1ex}
\addtolength{\abovedisplayshortskip}{-1ex}
\addtolength{\belowdisplayskip}{-1ex}
\addtolength{\belowdisplayshortskip}{-1ex}
}

\title{Sum Rate Maximization in the Constant Envelope  MIMO Downlink with the RZF Precoder}
		\author{Ferhad Askerbeyli, \textit{Graduate Student Member, IEEE}, Wen Xu, \textit{Senior Member, IEEE}, Josef A. Nossek, \textit{Life Fellow, IEEE} \thanks{ This paper was presented in part at  \emph{2023 IEEE 98th Vehicular
  Technology Conference (VTC2023-Fall)}.}%
	\thanks{F.~Askerbeyli is with  School of Computation, Information and Technology, Technical University of Munich (TUM), 80333 Munich, Germany and also with  Munich Research Center, Huawei Technologies Duesseldorf GmbH, 80992 Munich, Germany (e-mail: ferhad.askerbeyli@tum.de).}%
    \thanks{J.~A.~Nossek is with School of Computation, Information and Technology, Technical University of Munich (TUM), 80333 Munich, Germany (e-mail: josef.a.nossek@tum.de).}%
	\thanks{W. Xu is with Munich Research Center, Huawei Technologies Duesseldorf GmbH, 80992 Munich, Germany (wen.xu@ieee.org).}}
		
\begin{document}
	    
	\maketitle 
\begin{acronym}
\acro{MIMO}{multiple input multiple output}
\acro{CE}{constant envelope}
\acro{i.i.d.}{independent and identically distributed}
\acro{CSIT}{channel state information at the transmitter}
\acro{PA}{power amplifier}
\acro{ZF} {zero-forcing}
\acro{MRT}{maximum ratio transmission}
\acro{RZF}{regularized zero-forcing}
\acro{SQINR}{signal-to-quantization, interference and noise ratio}
\acro{SINR}{signal-to-interference and noise ratio}
\acro{SNR}{signal-to-noise ratio}
\acro{DAC}{digital-to-analog converter}
\acro{ADC}{analog-to-digital converter}
\acro{MUI}{multi-user interference}
\acro{CLT}{central limit theorem}
\acro{w.r.t.}{with respect to}
\acro{AQNM}{additive quantization noise model}
\acro{LCA}{linear covariance approximation}
\acro{MSE}{mean square error}
\acro{FLOP}{floating point operation}
\acro{RF}{radio frequency}
\acro{SISO}{single-input-single-output}
\acro{MMP}{mixed monotonic programming}
\acro{MM}{mixed monotonic}
\end{acronym}
\begin{abstract}
Feeding \acp{PA} with \ac{CE} signals is an effective way to reduce the power consumption in massive \ac{MIMO} systems. The nonlinear distortion caused by \ac{CE} signaling must be mitigated by means of signal processing to improve the achievable sum rates. To this purpose, many linear and nonlinear precoding techniques have been developed for the \ac{CE} \ac{MIMO} downlink. The vast majority of these \ac{CE} precoding techniques do not include a power allocation scheme, which is indispensable to achieve adequate performances in the downlink with channel gain imbalances between users. In this paper, we 
present two algorithms to produce a power allocation scheme for \ac{RZF} precoding in \ac{CE} \ac{MIMO} downlink. Both techniques are based on transforming the \ac{CE} quantized  \ac{MIMO} downlink to an approximately equivalent system of parallel \ac{SISO} channels. The first technique is proven to solve the sum rate maximization problem in the approximate system optimally, whereas the second technique obtains the local maximum with lower complexity.  We also extend another state-of-te-art  quantization aware sum rate maximization algorithm with linear precoding to the \ac{CE} downlink. Numerical results illustrate significant gains for the performance of the \ac{RZF} precoder when the \ac{CE} quantization is taken into account in a power allocation. Another key numerical result is that the proposed \ac{RZF} techniques achieve almost the identical performance so that the one with lower computational complexity is chosen as the main method. Results also show that the proposed \ac{RZF} precoding schemes perform at least as good as the state-of-the-art method with an advantage that the main \ac{RZF} method has significantly lower computational complexity than the state-of-the-art.   
\end{abstract}

\section{Introduction}
For 5G and beyond  mobile communications, massive \acf{MIMO} is a key technology  to achieve target performances in spectral efficiency, reliability  and coverage \cite{RuPeLaMa13,LaEdTuMa14,LuLiYe14}. Implementation of a fully digital massive \ac{MIMO} requires a \ac{RF} chain for each antenna element separately.  
As a result, fully digital  massive \ac{MIMO} system suffers from a low energy efficiency due to high number of active \ac{RF} chain components such as \acp{DAC}, \acp{ADC}, \acfp{PA},  mixers etc.

One approach to recover the energy efficiency is to reduce  power consumption 
of the \ac{RF} chain components, especially the components that contribute to the power consumption most. 
In the downlink, which is the scenario of our interest in this study,  the primary way to recover energy efficiency is to maximize the \ac{PA} efficiency, since \ac{PA} is the most power-hungry component in the \ac{RF} chain on the transmitter side \cite{BlZeBa10}. For linear \acp{PA}, the drain efficiency is maximized  by ensuring that \ac{PA} input signals  have a fixed magnitude in all channel uses, i.e., by feeding \acp{PA} with \acf{CE}  signals \cite{Ra86}. Furthermore, \ac{CE} signals enable use of nonlinear \acp{PA}, which are designed to operate at very high power efficiency. The secondary way is to reduce the  power consumption of \acp{DAC} by decreasing their resolution. Employing 1-bit \acp{DAC} serves both primary and secondary ways, as it minimizes the power consumption at the  digital-to-analog conversion while generating \ac{CE} input signals.

CE signaling and use of low resolution \acp{DAC}  come at a cost of severe quantization distortion that significantly deteriorates the transmit signal. Thus, many linear and nonlinear precoding techniques have been developed to suppress the quantization distortion and  achieve target performances in the quantized systems with high energy efficiency.  Symbol-wise nonlinear precoding methods especially have been successful in achieving solid data and error rates by suppressing the quantization distortion for every transmission individually \cite{JaDuCo17,NeStKr22,JeMeNo18}.

The above-mentioned quantized precoding techniques are designed for the downlink where users have a common large-scale fading coefficient and they fail in channels with varying large-scale fading coefficients for different users. In this case, a power allocation mechanism prioritizing between users is indispensable to achieve the satisfactory sum or error rate performances. An arising challenge in quantized systems is to develop a precoding technique that simultaneously handles the power allocation and the quantization distortion. 

\subsection{Related Works}
Power allocation for the \ac{MIMO} downlink with high resolution \acp{DAC}  has been studied extensively. For instance, power allocation for the weighted sum rate maximization in  high resolution downlink with \ac{ZF} precoding is a convex problem and it is solved by the well-known waterfilling algorithm \cite{YooGo06}. A broader class of linear precoding is defined by \acf{RZF}, where the channel inversion operation is controlled via a regularization parameter. Weighted sum rate maximization with the \ac{RZF} precoding is a problem of finding the optimal regularization parameter and the power allocation jointly and unlike  with the \ac{ZF} precoding,  the problem is not convex. Furthermore, analysis of the exact system is very difficult compared to the \ac{ZF} case.  In \cite{WaCo12},  a broad analysis of the MIMO downlink with \ac{RZF} precoding is provided by utilizing the  large system approximation. A power allocation scheme for \ac{MIMO} downlink with \ac{RZF} is also included in \cite{WaCo12}. However, the power allocation  in \cite{WaCo12} is not for users with different large-scale fading coefficients, but for users with different quality of \ac{CSIT}. Authors of \cite{MuZaEv13} derived joint optimality conditions for  power allocation, user loading and regularization to maximize the sum rate in MIMO downlink with the  \ac{RZF} precoding. Unlike \cite{WaCo12}, power allocation in \cite{MuZaEv13} is to tackle the channel gain imbalances between users. In \cite{SaBj14}, power allocation that minimizes the transmit power while satisfying individual \ac{SINR} constraints is obtained and examined via the large system approximation.

Studies on quantized systems with \ac{ZF} or \ac{RZF} precoding have also been reported. Authors of \cite{SaFiSwLe17} combined the large system approximation with the Bussgang decomposition to provide a performance analysis of \ac{MIMO} downlink with \ac{ZF} precoding and 1-bit \acp{DAC}. A similar analysis  for the \ac{MIMO} downlink  with \ac{ZF} precoder  and \ac{CE} quantization is done in \cite{SaMeHe20}, where the error rate performance is improved by introducing a Gaussian dither to the precoded signal. 
Analysis of MIMO downlink with 1-bit \acp{DAC} and \ac{RZF} precoding is done in \cite{XuXuGo19} to optimize the regularization parameter and the user loading.

None of  \cite{WaCo12,MuZaEv13,SaBj14,SaFiSwLe17,SaMeHe20,XuXuGo19} considered power allocation and quantization together. The first study with power allocation mechanism and quantization awareness is \cite{ChPaLe22}, where energy efficiency  maximization problem of the \ac{MIMO} downlink with low-resolution \acp{DAC} is solved by optimizing all elements of a linear precoding matrix jointly.  A special case of energy efficiency optimization problem in \cite{ChPaLe22} is the sum rate maximization,  which implicitly includes a power allocation problem as well. In contrast to \cite{ChPaLe22}, a quantization aware power allocation scheme, which aims at sum rate maximization in the 1-bit \ac{MIMO} downlink with \ac{ZF} precoding, is explicitly formulated with a power factor per each user in \cite{AsXuNo23}.

\subsection{Main Contributions}
Employing \ac{RZF} precoders with power allocation  in \cite{MuZaEv13} and \cite{SaBj14}  for \ac{CE} quantized systems is clearly suboptimal, as they disregard the quantization distortion. Previously in \cite{AsXuNo23_2}, we handled this suboptimality by proposing 1-bit quantization aware power allocation for  \ac{RZF} precoding. In this paper,
we extend the heuristic method in \cite{AsXuNo23_2} to the \ac{CE} downlink with higher resolution and provide a novel method that solves the same problem as the heuristic method optimally. Contributions of this work are summarized as follows: 
\begin{itemize}
\item An asymptotic analysis of a \ac{MIMO} downlink with \ac{RZF} precoding and \ac{CE} transmit signals is provided by combining the large system approximation with high transmit power assumption. The asymptotic analysis leads to derivation of an approximately equivalent system with parallel \acf{SISO} channels. 
\item Two algorithms that obtain a power allocation and a regularization parameter  to maximize the sum rate in the approximate system are proposed. One of the algorithms with branch and bound method obtains the global maximum for the approximate system. The other  algorithm with alternating optimization converges to a local maximum. 
\item Behavior of the alternating algorithm at high transmit power regime  and its computational complexity are analyzed in detail.
\item The state-of-the-art quantization aware linear precoding method in \cite{ChPaLe22} is nontrivially extended for the \ac{CE} quantization case. Later it is used for comparison with the proposed \ac{RZF} precoding techniques.  
\end{itemize}
The asymptotic analysis and the alternating algorithm are presented in \cite{AsXuNo23_2}, the other contributions (including the analysis of the alternating algorithm) are introduced for the first time.

Two most relevant works to ours are \cite{MuZaEv13} and \cite{XuXuGo19}, where the former excludes the \ac{CE} quantization and the latter excludes  the power allocation. Both \cite{MuZaEv13} and \cite{XuXuGo19} mainly focus on user-loading problem and the latter considers only 1-bit quantization. This work considers a system in the intersection of the settings in \cite{MuZaEv13} and \cite{XuXuGo19}, i.e., it is  a nontrivial generalization of both regarding the power allocation problem. 

On the contrary to the quantization aware generalized power iterations for spectral efficiency maximization (Q-GPI-SEM) algorithm from  \cite{ChPaLe22}, we impose the \ac{RZF} structure on the linear precoding matrix. By doing so, we reduce the joint optimization of the whole precoding matrix to the joint optimization of power factors of users, regularization parameter and number of users to serve, which reduces the computational complexity. Furthermore, the methods we propose are compatible with any  input constellation, whereas Q-GPI-SEM's performance deteriorates if the input signal is not Gaussian.  

\subsection{Remainder and Notation}
The structure of this paper is as follows:  System model  is introduced in  \Cref{sec:sys_mod}. The approximate \acp{SQINR} of users are computed in \Cref{sec:asymptotic_analysis} by utilizing the asymptotic analysis.  An approximately equivalent system of parallel \ac{SISO} channels based on the approximate \acp{SQINR} is presented in \Cref{sec:approx_sys}.  Two algorithms that perform sum rate maximization for this approximate system are devised in \Cref{{sec:opt_alg}}. A high transmit power  and complexity  analysis of the alternating algorithm  is presented in \Cref{sec:alter_alg}. The Q-GPI-SEM algorithm is extended to \ac{CE} quantized systems in  \Cref{sec:ex_QA_GPI_SEM}. Numerical  results and conclusions are reported in \Cref{sec:num_res} and \Cref{sec:concl}, respectively.

\textit{Notation}: The $m$th entry of vector $\cvec{a}$ and the $(m,n)$th entry of  matrix $\bm{A}$ are  denoted as $a_{m}$ or $[\cvec{a}]_m$ and $a_{m,n}$ or $[\bm{A}]_{m,n}$, respectively.  Expressions $\text{diag}(\cvec{A})$, $\text{diag}(\cvec{a})$ and $\text{diag}(a_1,a_2,\ldots a_n)$  stand for  diagonal matrices  with the entries of the diagonal of matrix $\cvec{A}$,  vector $\cvec{a}$ and vector $[a_1,a_2, \ldots a_n]$, respectively.  The covariance matrix between vectors $\cvec{a}$ and $\cvec{b}$ is denoted as $\bm{C}_{\cvec{a}\cvec{b}}$. A circularly symmetric complex Gaussian distribution with the  mean $\cvec{m}$ and covariance $\bm{C}$ is denoted as $\mathcal{CN}(\cvec{m},\bm{C})$. The term $\text{vec}(\bm{A})$ denotes the vector obtained by stacking columns of $\bm{A}$. The closed interval between $a$ and $b$ is denoted as $[a,b]$.
\section{System Model and Problem Formulation} \label{sec:sys_mod}
	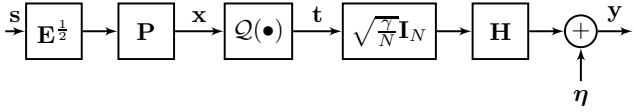
\begin{figure}[h]   
	\centering
	\resizebox{\columnwidth}{!}{\tikzstyle{int}=[draw, minimum width=0.8cm, minimum height=0.8cm, thick]
\tikzstyle{init} = [pin edge={<-,thick,black}]
\tikzstyle{sum} = [draw, circle,inner sep=1pt, minimum size=2mm,  thick] 
\begin{tikzpicture}[node distance=1cm,auto,>=latex']
    \node [int] (a) {$\bm{E}^{\frac{1}{2}}$};
    \node (b) [left of=a,node distance=0.7cm, coordinate] {a};
    \node [int] (c1) [right=0.5cm of a] {$\bm{P}$};
    \node [int] (c2) [right=0.7cm of c1] {$\mathcal{Q}(\bullet)$};
    
    \node [int] (c) [right=0.7cm of c2] {$\sqrt{\frac{\gamma}{N}}\mathbf{I}_N$};
    \node [int] (f) [right=0.5cm of c] {$\bm{H}$};
    \node [sum,  pin={[init]below:$\boldsymbol{\eta}$}] (g) [right=0.5cm of f] {$\mathbf{+}$};
    \node [coordinate] (end) [right=0.5cm of g, node distance=0.6cm]{};

    \path[->,thick] (b) edge node[above] {$\cvec{s}$} 
                       node[below] {} (a);
    \path[->,thick] (c) edge (f);
    \path[->,thick] (f) edge (g);
    \path[->,thick] (a) edge (c1);
    \path[->,thick] (c1) edge node[above] {$\cvec{x}$} node[below]{}  (c2);
    \path[->,thick] (c2) edge node[above]{$\cvec{t}$} node[below] {}  (c); 
    \path[->,thick] (g) edge node[above] {$\mathbf{y}$} 
                       node[below] {} (end);                        
\end{tikzpicture}}
	\caption{MIMO downlink with linear precoding and \ac{CE} quantization.}
	\label{fig:downlink_quan}  
	\end{figure}

Fig.~\ref{fig:downlink_quan} illustrates  the \ac{MIMO} downlink of our interest, which consists of a transmitter with $N$ antennas and $M$ single-antenna users  with $M\leq N$.  The channel between the transmitter and single-antenna users is modeled as  $\bm{H}=\bm{\Sigma}^{\frac{1}{2}}\tilde{\bm{H}}\in\mathbb{C}^{M \times N}$, where $\bm{\Sigma}=\text{diag}(\sigma_1,\ldots, \sigma_M)$ contains the large-scale fading coefficients of all users in nonincreasing order such that  $\sigma_1\geq \sigma_2 \geq \ldots \geq \sigma_M$    and  matrix $\tilde{\bm{H}}\in\mathbb{C}^{M \times N}$ consists of uncorrelated elements with $[\tilde{\bm{H}}]_{m,n}\sim\mathcal{CN}(0,1)\,\forall m,n$. We denote the $m$th row of $\bm{H}$ and $\tilde{\bm{H}}$ as $\cvec{h}_m^{\op{H}}$ and $\tilde{\cvec{h}}_m^{\op{H}}$, respectively and $\cvec{h}_m^{\op{H}}=\sqrt{\sigma_m}\tilde{\cvec{h}}_m^{\op{H}}$.

The input signal $\cvec{s}\in\mathbb{C}^M$ consists of zero mean,  \ac{i.i.d.} symbols with unit variance, i.e., $\operatorname{E}[\cvec{s}]=\cvec{0}$ and $\op{E}[\cvec{s}\cvec{s}^{\op{H}}]=\bm{0}$, that are to be transmitted to the corresponding user. There is no restriction other  than all input symbols have a common constellation and distribution. 

A two-stage linear precoding scheme is employed. In the first stage, the input signal $\cvec{s}$ is weighted by $\bm{E}^{\frac{1}{2}}$, where $\bm{E}=\text{diag}(
e_1,e_2,\ldots e_M)$ is the power allocation matrix and $e_m\geq 0$ is the power factor of the $m$th user. In the second stage,  mapping of $\cvec{s}$ to the precoded signal $\cvec{x}$ is carried out as $\cvec{x}=\bm{P}\bm{E}^{\frac{1}{2}}\cvec{s}$, where $\bm{P}$ is the precoding matrix determined according to the \ac{CSIT}. Power factors in  $\bm{E}$   set the received signal powers at users and should be tuned carefully to achieve high sum rates especially in channels with various large-scale fading coefficients for different users. As some users may have very weak channels, it may be optimal to allocate power only to a subset of users with the strongest channels such that $e_1,e_2,\ldots e_K>0$ and $e_{K+1}=e_{K+2}\ldots=e_M=0$.  In such cases, symbols of users with indices from $K+1$ to $M$ do not contribute to the precoded signal, i.e., they do not get served. 

For the second stage, we consider the \ac{RZF} precoding which reads as follows when all $M$ users are taken into account
\begin{equation}\label{eq:rzf_full}
\bm{P}=(\tilde{\bm{H}}^{\op{H}}\tilde{\bm{H}}+\alpha\eye_N)^{-1}
\tilde{\bm{H}}^{\op{H}}\bm{\Sigma}^{-\frac{1}{2}}.
\end{equation}
Equation \eqref{eq:rzf_full} is  a controlled inversion of the small-scale fading matrix $\tilde{\bm{H}}$ with a nonnegative regularization parameter $\alpha$ and inversion of the large-scale fading matrix $\bm{\Sigma}^{\frac{1}{2}}$. As power allocation may cause only a subset of users to get served, the \ac{RZF} precoding  does not always need to take all users into account. For that reason, we define the \ac{RZF} precoder 
depending on the maximum number of users we aim to serve. The \ac{RZF} precoding  we employ for $K$ users reads as    
\begin{equation}\label{eq:rzf}
\bm{P}=\begin{bmatrix}
\bm{P}_K &\bm{0}
\end{bmatrix}=\begin{bmatrix}
(\tilde{\bm{H}}_{K}^{\op{H}}\tilde{\bm{H}}_{K}+\alpha\eye_N)^{-1}
\tilde{\bm{H}}_{K}^{\op{H}}\bm{\Sigma}_K^{-\frac{1}{2}}& \bm{0}
\end{bmatrix},
\end{equation} 
where $\tilde{\bm{H}}_K\in\mathbb{C}^{K\times N}$ consists of first $K$ rows of $\tilde{\bm{H}}$ and $\bm{\Sigma}_K=\text{diag}(\sigma_1,\sigma_2, \ldots \sigma_K)$. The \ac{RZF} in \eqref{eq:rzf} leads to the following precoding
\begin{equation}
\label{eq:precoded_signal}
\cvec{x}=\bm{P}_K\bm{E}^{\frac{1}{2}}_K\cvec{s}_{K}=\sum_{k=1}^{K}\sqrt{e_k}\cvec{p}_k s_k,
\end{equation}
where $\bm{E}_K=\text{diag}(e_1,e_2,\ldots e_K)$ and $\cvec{s}_K=[s_1  s_2  \ldots s_K]^{\op{T}}$ and $\cvec{p}_k$ is the $k$th column of $\bm{P}_K$. As a result, design of the two-stage \ac{RZF} precoder consists of jointly determining number of users to serve, power allocation for the served users   and the regularization parameter. Note that  \ac{RZF} precoding in \eqref{eq:rzf}  includes \eqref{eq:rzf_full} as a special case of $K=M$. 

Elements of the precoded signal $\cvec{x}\in \mathbb{C}^{N}$ go through the following \ac{CE} quantization with $Q$ levels
\begin{equation}\label{eq:CE_quantization}
t_n= \mathcal{Q}(x_n)=\exp(\text{j}(\ceil{\frac{\angle{x_n}}{2\psi}}-\psi)), \forall n, 
\end{equation}     
where $\psi=\frac{\pi}{Q}$, i.e., every precoded symbol is mapped to one of the $Q$ discrete points with unit magnitude. We also use the vector notation for this element-wise quantization as $\cvec{t}=\mathcal{Q}(\cvec{x})$. Note that the power of quantized signal $\cvec{t}$ is equal to $N$ and \ac{CE} quantization is invariant to scaling of $\cvec{x}$ with a positive constant $v$ such that  $\mathcal{Q}(\cvec{x})=\mathcal{Q}(v\cvec{x})$.  
At \acp{PA},  the quantized signal $\cvec{t}$ is scaled by factor of $\sqrt{\frac{\gamma}{N}}$ to  produce a transmit signal with power of $\gamma$. The transmit signal $\cvec{t}$ propagates through channel $\bm{H}$ and is received at single-antenna users with additive white Gaussian noise (AWGN) of $\bs{\eta} \sim \mathcal{CN}(\cvec{0}_M,\eye_M)$ as follows
\begin{equation}
\label{eq:io}
\cvec{y}=\sqrt{\frac{\gamma}{N}}\bm{H}\cvec{t}+\bs{\eta}.
\end{equation}
 
In this paper, our primary goal is to obtain a two-stage \ac{RZF} precoding scheme that maximizes the sum rate of the \ac{CE} quantized downlink described in Fig.~\ref{fig:downlink_quan}.  Since two-stage \ac{RZF} precoding can be parametrized by $K$, $\bm{E}_K$ and $\alpha$,  the sum rate maximization problem can be formulated as follows
\begin{equation}\label{eq:original_prob_with_K}
\underset{\substack{\alpha \geq 0,\, K\in\{1,2,\ldots M\} \\e_1,e_2,\ldots e_K\geq0}}{\max} \sum_{k=1}^{K} I(y_k;s_k), 
\end{equation} 
where $I(y_k;s_k)$ is the mutual information between the $k$th input and and received signal.

Solving \eqref{eq:original_prob_with_K} directly is challenging, since derivation of  an analytic expression  for the rate of the $k$th user $I(y_k;s_k)$ is  difficult as the distribution of $\cvec{s}$ is arbitrary and \ac{CE} quantization takes place. We instead approximately compute \ac{SQINR} values of the system  and then identify an equivalent system of parallel SISO channels with the help of computed \acp{SQINR}. In the end, the sum rate maximizing algorithm is developed based on the approximately equivalent system of parallel SISO channels. 





\section{Computation of \ac{SQINR}}\label{sec:asymptotic_analysis}

In this section, we compute \acp{SQINR} of the system described in \Cref{sec:sys_mod},  when \ac{RZF} precoding  is done for $K$ users. To this aim, we need to decompose the received signal at the $k$th user into uncorrelated components and compute the powers of these components. 
\subsection{Received Signal Decomposition}	

A linear relationship between $s_k$ and  $y_k$ is not immediately available because of the   the \ac{CE} quantization. Bussgang decomposition is conventionally applied to formulate the \ac{CE}  quantization  in \eqref{eq:CE_quantization}  as a linear stochastic process to circumvent the nonlinearity as follows  \cite{JaDuCo17,SaFiSwLe17}
\begin{equation}\label{eq:bussgang_decom}
\cvec{t}=\bm{B}\cvec{x}+\cvec{d}=\bm{B}\bm{P}_K\bm{E}_K^{\frac{1}{2}}\cvec{s}_K+\cvec{d},
\end{equation}
where $\bm{B}$ is the Bussgang gain matrix and $\cvec{d}$ is the distortion vector \cite{Bussgang52}. Bussgang gain matrix $\bm{B}$ for CE quantization of $\cvec{x}$ with $Q$ levels as $\cvec{t}=\mathcal{Q}(\cvec{x})$ is computed  in \cite{HeNo18} as 
\begin{equation}\label{eq:bussgang_mat_ce}
\bm{B}=\xi_Q\text{diag}(\bm{C}_{\cvec{x}\cvec{x}})^{-\frac{1}{2}}, 
\end{equation}
with 
\begin{equation}\label{eq:xi_Q}
\xi_Q=\frac{Q}{2\sqrt{\pi}}\sin(\frac{\pi}{Q}).
\end{equation} 
Note that $\xi_Q$ solely depends on the number of quantization levels $Q$ and it increases as $Q $ increases and   $\lim_{Q\to\infty}\xi_Q=\sqrt{\frac{\pi}{4}}$. For our case, the key property of Bussgang decomposition is that if $\cvec{x}$ is Gaussian distributed, then $\cvec{t}$ is decomposed into two uncorrelated components  $\cvec{d}$ and  $\bm{B}\cvec{x}$. The decomposition in \eqref{eq:bussgang_decom} can still be used if $\cvec{x}$ is not Gaussian, however, in this case $\cvec{d}$ would not be uncorrelated with $\cvec{x}$. Note that $\cvec{d}$ is certainly not Gaussian so that assuming Gaussian $\cvec{d}$ at the receivers would lead to mismatched decoding. 

 By combining  \eqref{eq:bussgang_mat_ce}, \eqref{eq:bussgang_decom}, \eqref{eq:io} and  \eqref{eq:precoded_signal}, one can  write the received signal for the $k$th user as 	
\begin{equation}\label{eq:io_decom}
\begin{aligned}
&y_k=\sqrt{\frac{\gamma}{N}e_k}\xi_Q \cvec{h}_k^{\op{H}}\text{diag}(\bm{C}_{\cvec{x}\cvec{x}})^{-\frac{1}{2}}\cvec{p}_k s_k+\sqrt{\frac{\gamma}{N}}\cvec{h}_{k}^{\op{H}}\cvec{d}\\&+\sqrt{\frac{\gamma}{N}}\xi_Q \cvec{h}_k^{\op{H}}\text{diag}(\bm{C}_{\cvec{x}\cvec{x}})^{-\frac{1}{2}}\sum_{\substack{j\neq k}}\sqrt{e_j}\cvec{p}_j s_j+\eta_k,
\end{aligned} 
\end{equation} 
for  $k =1,2,\ldots K$. In \eqref{eq:io_decom}, the received signal is decomposed  as the component that is linearly dependent on input $s_k$, the received quantization distortion,  the \ac{MUI} and the additive noise, respectively. Yet, computing the exact \ac{SQINR} of the $k$th user is still very difficult by \eqref{eq:io_decom} since the quantization distortion vector $\cvec{d}$ is  not uncorrelated from the input signal $\cvec{s}$, unless $\cvec{x}$ is Gaussian.  Furthermore, computing the exact power of the linearly dependent component, \ac{MUI} and quantization distortion  in terms of regularization parameter $\alpha$ is very difficult. For those reasons, we resort to asymptotic approximation.  

\subsection{Asymptotic Approximation}\label{sec:asympto}
The difficulty we encounter in computing the exact \ac{SQINR} values of the system described in Fig.~\ref{fig:downlink_quan} is present in \cite{MuZaEv13} and  \cite{XuXuGo19}, where  researchers resorted to asymptotic approximation, as we also follow here.

Asymptotic approximation is typically used interchangeably with the so-called "large system approximation", which  consists of approximating  metrics of a \ac{MIMO} system under the assumption that $M\to\infty, N\to\infty$ with a fixed user load $\beta= \frac{M}{N}$. Large system approximation is based on one of the key merits of massive \ac{MIMO} known as channel hardening. As  dimensions of a \ac{MIMO} system get large, the system effectively acts more deterministic, i.e.,  \ac{SINR}/\ac{SQINR} at the receivers do not depend on  the channel realization $\bm{H}$, but they depend on the channel dimensions $M$, $N$.  For any system with a finite $M, N$ and a given $\beta$, an approximate system based on the large system approximation can be obtained and used to develop algorithms for the original system. 

In our  scenario, we assume a more realistic channel model including large-scale fading coefficients and thus  power factor matrix $\bm{E}_K$ is also involved in precoding. For that reason,  the typical large system approximation is not sufficient  to characterize the overall system in \eqref{eq:io_decom}. To overcome this problem, we additionally assume that number of users $K$ for which \ac{RZF} precoding is designed is large and 
\begin{equation}\label{eq:high_transmit_power}
\bm{E}_K\bm{\Sigma}_K^{-1}\approx \frac{\op{tr}(\bm{E}_K\bm{\Sigma}_K^{-1})}{K} \eye_K. 
\end{equation} In \Cref{sec:high_tr}, we show  that these assumptions get more accurate at the high transmit power regime for the algorithms we present. Thus, we refer to these two additional assumptions together  as the high transmit power  assumption. Note that \eqref{eq:high_transmit_power} is not imposed  as a constraint when optimizing $\bm{E}_K$, but it serves to only identify an approximate system which get more accurate at high transmit power regime.
\subsubsection{Large System Approximation}\label{subsec:random}
Let us first recap few identities from random matrix theory that we repeatedly use for large system approximation.
\begin{lemma}
(Corolllary 1 in \cite{EvTs00}) Let $\bm{A}$ be a deterministic $N\times N$ complex matrix with  bounded spectral radius for all $N$. Let $\tilde{\cvec{h}}_k \in \mathbb{C}^{N\times 1}$ consist of  i.i.d complex random variables with zero mean, unit variance and finite eight moment. Then,
\begin{equation}\label{eq:corollary1}
\frac{1}{N}\tilde{\cvec{h}}_k^{\op{H}}\bm{A}\tilde{\cvec{h}}_k\approx\frac{1}{N}\op{tr}(\bm{A}).
\end{equation}
\end{lemma}
\begin{theorem}\label{th:theorem1}
(Theorem~7 in \cite{EvTs00})  For $\tilde{\bm{H}}_K\in\mathbb{C}^{K\times N}$ consisting of i.i.d elements $[\tilde{\bm{H}}_K]_{m,n}\sim\mathcal{CN}(0,1)$, when  $K\to \infty$, $N\to\infty$ with fixed ratio $\beta=\frac{K}{N}$, it holds that
\begin{equation}\label{eq:8}
\op{tr}({({\tilde{\bm{H}}_K}^{\op{H}}\tilde{\bm{H}}_K+\alpha \eye_{N})}^{-1})\approx g(\beta,\rho),
\end{equation}
where $g(\beta,\rho)$ satisfies 
\begin{equation}\label{eq:9}
g(\beta,\rho)={\bigg(\rho+\frac{\beta}{1+g(\beta,\rho)}\bigg)}^{-1}
\end{equation}
	with the normalized regularization parameter $\rho=\frac{\alpha}{N}$. Notice  that $g(\beta,\rho)$ is positive by definition, since it is an asymptotic value of  trace of a positive semidefinite matrix. The positive solution of \eqref{eq:9} is given as 
	\begin{equation}\label{eq:g_pos}
	g(\beta,\rho)=\frac{\sqrt{{(\beta+\rho-1)}^2+4\rho}-(\beta+\rho-1)}{2\rho}.
	\end{equation} 
For the rest of this paper, let us introduce $\bm{M}_K={({\tilde{\bm{H}}_K}^{\op{H}}\tilde{\bm{H}}_K+\alpha \eye_{N})}^{-1}$ to ease the notation. 
\end{theorem} 		
	
\textbf{Corollary 1.} (\cite{MuZaEv13}) The following approximation holds under the same conditions as  Theorem ~\ref{th:theorem1}
\begin{equation}\label{eq:14}
-\alpha\op{tr}(\bm{M}_K^{2})\approx\rho \frac{\partial g(\beta,\rho)}{\partial\rho}.
\end{equation}
Corollary 1 can be easily derived by taking the derivative of \eqref{eq:8} with respect to $\alpha$ on both sides. 

	
By applying simple algebraic manipulations on	\eqref{eq:9}, we can obtain the following identities
\begin{equation}\label{eq:10}
\beta+\rho{(1+g(\beta,\rho))}^2=\frac{(1+g(\beta,\rho))^2-\beta g^2(\beta,\rho)}{g(\beta,\rho)}
\end{equation}
and 
\begin{equation}\label{eq:11}
\begin{aligned}
&\frac{\partial g(\beta,\rho)}{\partial \rho}=-\frac{g(\beta,\rho){(1+g(\beta,\rho))}^2}{\beta+\rho{(1+g(\beta,\rho))}^2}.
\end{aligned}
\end{equation}
For a given $\beta$, equation \eqref{eq:11} implies that  $g(\beta,\rho)$ is strictly decreasing in $\rho$, since $g(\beta,\rho)>0$ and $\rho\geq0$. Furthermore, one can show that
\begin{equation}\label{eq:eq_lim} 
\lim_{\rho \to \infty} g(\beta,\rho)= 0
\end{equation}	
with a simple manipulation of \eqref{eq:g_pos}.

\subsubsection{An Alternative Regularization Parameter}
Let us introduce the following term
\begin{equation}\label{eq:def_u}
 u(\beta,\rho)=\frac{g(\beta,\rho)}{1+g(\beta,\rho)}, 
\end{equation}
for which it is implied that $\lim_{\rho\to \infty}u(\beta,\rho)=0$ due to  \eqref{eq:eq_lim}. By manipulating \eqref{eq:9} and \eqref{eq:g_pos}, we compute that $\lim_{\rho\to 0}u(\beta,\rho)=1$ for $\beta\leq 1$. 
The derivative of \eqref{eq:def_u} yields
\begin{equation}\label{eq:deriv_u}
\frac{\partial u(\beta,\rho)}{\partial \rho}= \frac{1}{(1+g(\beta,\rho))^2}\frac{\partial g(\beta,\rho)}{\partial \rho},
\end{equation}
which  implies that $u(\beta,\rho)$ is  strictly decreasing in $\rho$, since from \eqref{eq:11} we inferred  that $\frac{\partial g(\beta,\rho)}{\partial\rho}<0$. Hence, for a fixed value of $\beta$,  $u(\beta,\rho)$ starts as $1$ at $\rho=0$  and strictly decreases to $0$ as $\rho\to\infty$. We hence can use $u$ as an alternative regularization parameter with a one-to-one relationship to $\alpha$ 
\begin{equation}\label{eq:rho}
\alpha(u)=N(\frac{1}{u}-\beta)(1-u),
\end{equation}  
which is obtained by combining \eqref{eq:9}, \eqref{eq:def_u} and definition of $\rho=\frac{\alpha}{N}$. Note that we denoted $u(\beta,\rho)$ as $u$ for the sake of brevity, as we do for the rest of this paper. From \eqref{eq:rho}, we observe that as \ac{RZF} precoder approaches to \ac{ZF} precoder  $u$ approaches to $1$ and as    \ac{RZF} precoder approaches to \ac{MRT} precoder, $u$ approaches to 0.

In the rest of the current section, the approximate power of the desired signal component and \ac{MUI} is expressed in terms of $u$. On the contrary to the other regularization parameter $\rho$, which is accompanied by $g(\beta,\rho)$ in signal  or \ac{MUI} power expressions (see \cite{MuZaEv13,XuXuGo19}), $u$ solely reflects  the effect of regularization on the approximate power of the desired signal component and MUI. This makes use of $u$ more practical as it leads to more intuitive and compact expressions.



	
	

\subsection{Application of Asymptotic Approximation}\label{sec:approx_SQINR}
By combining $\text{diag}(\bm{C}_{\cvec{x}\cvec{x}})=\text{diag}(\bm{P}_{K}\bm{E}_K\bm{P}_K^{\op{H}})$ and $\bm{P}_K=\bm{M}_K\tilde{\bm{H}}_K^{\op{H}}\bm{\Sigma}_K^{-\frac{1}{2}}$, we  apply  the asymptotic approximation to  $\text{diag}(\bm{C}_{\cvec{x}\cvec{x}})$  as  follows 
\begin{equation}\label{eq:pd_init}
\begin{split}
 \text{diag}(\bm{C}_{\cvec{x}\cvec{x}})&=\text{diag}(\bm{M}_K{\tilde{\bm{H}}_K}^{\op{H}}\bm{E}_K\bm{\Sigma}^{-1}_K\tilde{\bm{H}}_K\bm{M}_K)\\
&\overset{(a)}{\approx} \frac{\op{tr}(\bm{E}_K\bm{\Sigma}_K^{-1})}{K} \text{diag}(\bm{M}_K{\tilde{\bm{H}}_K}^{\op{H}}\tilde{\bm{H}}_K\bm{M}_K)\\
&\overset{(b)}{=}\frac{\op{tr}(\bm{E}_K\bm{\Sigma}_K^{-1})}{K} \text{diag}(\bm{M}_K-\alpha\bm{M}_K^{2})\\&\overset{(c)}{\approx}\frac{\op{tr}(\bm{E}_K\bm{\Sigma}_K^{-1})}{K}\bigg(\frac{\op{tr}(\bm{M}_K)}{N}-\alpha\frac{\op{tr}(\bm{M}_K^2)}{N}\bigg)\eye_N.
\end{split}
\raisetag{50pt}
	\end{equation} 
Step $(a)$ is taken by employing \eqref{eq:high_transmit_power}, step $(b)$ is taken by employing  ${\tilde{\bm{H}}_K}^{\op{H}}\tilde{\bm{H}}_K\bm{M}_K= \eye_{N}-\alpha \bm{M}_K$, step $(c)$ is by approximating   $\text{diag}(\bm{M}_K)\approx\frac{\op{tr}(\bm{M}_K)}{N}\eye_N$, when $K$ is large.

We proceed on approximating $\text{diag}(\bm{C}_{\cvec{x}\cvec{x}})$ by replacing \eqref{eq:8} and \eqref{eq:14} in \eqref{eq:pd_init} and then making  use of \eqref{eq:10} and \eqref{eq:11} 
	\begin{equation}\label{eq:15}
	\begin{aligned}
	\text{diag}(\bm{C}_{\cvec{x}\cvec{x}})&\approx \frac{\op{tr}(\bm{E}_K\bm{\Sigma}_K^{-1})}{KN} (g(\beta,\rho)+\rho \frac{\partial g(\beta,\rho)}{\partial\rho})\eye_N\\
	&= \frac{\op{tr}(\bm{E}_K\bm{\Sigma}_K^{-1})}{KN}\frac{\beta u^2}{(1-\beta u^2)} \eye_N.
	\end{aligned}
	\end{equation}

\subsubsection{Power of the Desired Signal}	
By combining \eqref{eq:15} and $\cvec{p}_k=\frac{1}{\sqrt{\sigma_k}}\bm{M}_K\tilde{\bm{h}}_{k}$, desired signal at the $k$th user  in \eqref{eq:io_decom} is approximated as
\begin{equation}
y^{(\text{s})}_k\approx\sqrt{\frac{\gamma w_k \sigma_k (1-\beta u^2)}{\beta u^2}}\xi_Q\tilde{\cvec{h}}^{\op{H}}_{k} \bm{M}_K\tilde{\cvec{h}}_ks_k,
\end{equation}
where we also introduced 
\begin{equation}\label{eq:w_k_def}
w_k=\frac{e_k/\sigma_k}{\op{tr}(\bm{E}_K\bm{\Sigma}_K^{-1})}K\geq 0\;\text{for}\; k=1,2,\ldots K.
\end{equation}
The following approximation can be derived by following the same steps as in Appendix A of \cite{MuZaEv13}, where matrix inversion lemma has been utilized
\begin{equation}
|\tilde{\cvec{h}}^{\op{H}}_{k}\bm{M}_K\tilde{\cvec{h}}_k|^{2}\approx \frac{g^2(\beta,\rho)}{(1+g(\beta,\rho))^2}=u^2.
\end{equation}

As a result, the signal power can be approximated as
\begin{equation}\label{eq:desired_fin}
\begin{aligned}
P_k^{(\text{s})}\approx\xi_Q^2\gamma\sigma_k w_k \bigg(\frac{1}{\beta}-u^2\bigg),
\end{aligned}
\end{equation}	
In \eqref{eq:desired_fin}, increasing $u$ leads to decrease in signal power.
	\subsubsection{Power of the \ac{MUI}}
	Let us now approximate the \ac{MUI} signal as follows 
	\begin{equation}
	\begin{aligned}
	&y^{(\text{mui})}_k=\sqrt{\frac{\gamma}{N}}\xi_Q\cvec{h}_k^{\op{H}}{\text{diag}(\bm{C}_{\cvec{x}\cvec{x}})}^{-\frac{1}{2}}\sum_{\substack{j\neq k}}\sqrt{\frac{e_j}{\sigma_j}}\bm{M}_K\tilde{\cvec{h}}_{j} s_j\\
	&\approx \sqrt{\frac{\gamma}{N}\frac{\sigma_k}{\frac{\op{tr}(\bm{E}_K\bm{\Sigma}_K^{-1})}{KN}\frac{\beta u^2}{1-\beta u^2}}} \xi_Q\sum_{\substack{j\neq k}} \sqrt{\frac{e_j}{\sigma_j}}\tilde{\cvec{h}}_k^{\op{H}}\bm{M}_K \tilde{\cvec{h}}_j s_j.
	\end{aligned}
	\end{equation}
Power of the \ac{MUI} is then approximated as 
	\begin{equation}\label{eq:p_mui_init}
	P^{(\text{mui})}_k\approx \xi_Q^2\frac{\gamma}{N}\frac{\sigma_k}{\frac{\op{tr}(\bm{E}_K\bm{\Sigma}_K^{-1})}{KN}\frac{\beta u^2}{1-\beta u^2}} \sum_{\substack{j\neq k}}^{K} \frac{e_j}{\sigma_j} {|\tilde{\cvec{h}}_k^{\op{H}}\bm{M}_K\tilde{\cvec{h}}_j|}^2
	\end{equation}	
	In \cite{MuZaEv13}, the term  ${|\tilde{\cvec{h}}_k^{\op{H}}\bm{M}_K\tilde{\cvec{h}}_j|}^2$ is approximated as
\begin{equation} \label{eq:mui_term}
|\tilde{\cvec{h}}_k^{\op{H}}\bm{M}_K\tilde{\cvec{h}}_j|^2\approx \frac{1}{N} \frac{g(\beta,\rho)}{{(1+g(\beta,\rho))}^2(\beta+\rho(1+g(\beta,\rho))^2)}. 
\end{equation}
By plugging in  \eqref{eq:mui_term} and \eqref{eq:10} to \eqref{eq:p_mui_init}, th \ac{MUI} power  can be written as 
	\begin{equation}\label{eq:P_mui_2}
	P^{(\text{mui})}_k\approx \xi_Q^2\gamma\sigma_k (1-u)^2 \frac{1}{K} \sum_{\substack{j\neq k}} w_j
	\end{equation}
By employing the definition of $w_k$ in \eqref{eq:w_k_def} and the high transmit power assumption in \eqref{eq:high_transmit_power}, we can write 
\begin{equation}\label{eq:wsum_approx}
\frac{1}{K}\sum_{\substack{j\neq k}}^{K} w_j\approx 1- \frac{1}{K}.
\end{equation}     
Eventually, power of the \ac{MUI} at the $k$th user reads as
\begin{equation}\label{eq:P_mui}
P^{(\text{mui})}_k\approx \xi_Q^2\gamma\sigma_k\frac{K-1}{K} (1-u)^2.
\end{equation}      
In \eqref{eq:P_mui}, increase in $u$ leads to decrease  in \ac{MUI} power.
\subsubsection{Power of the Received Quantization Distortion }
Power of the quantization distortion at the $k$th user in \eqref{eq:io_decom} is given as  
\begin{equation}
P^{(\text{q})}_k= \frac{\gamma}{N}\cvec{h}_k^{\op{H}}\bm{C}_{\cvec{d}\cvec{d}}\cvec{h}_k=\frac{\gamma\sigma_k}{N}\tilde{\cvec{h}}_k^{\op{H}}\bm{C}_{\cvec{d}\cvec{d}}\tilde{\cvec{h}}_k,
\end{equation}
where $\bm{C}_{\cvec{d}\cvec{d}}$ is the covariance matrix of the quantization distortion $\cvec{d}$. Except the case of 1-bit quantization i.e., CE quantization with $Q=4$, it is difficult to compute $\bm{C}_{\cvec{d}\cvec{d}}$ exactly. For that reason, we again resort to the asymptotic analysis. Lemma 1 in   \eqref{eq:corollary1} implies that 
\begin{equation}
\frac{1}{N}\tilde{\cvec{h}}_k^{\op{H}}\bm{C}_{\cvec{d}\cvec{d}}\tilde{\cvec{h}}_k\approx\frac{1}{N}\op{tr}(\bm{C}_{\cvec{d}\cvec{d}}),
\end{equation}
which leads to the following approximation
\begin{equation}
P^{(\text{q})}_k\approx\frac{\gamma\sigma_k}{N}\op{tr}(\bm{C}_{\cvec{d}\cvec{d}}).
\end{equation}
Recall that asymptotic approximation assumptions include  $K$ being large so that the precoded signal $\cvec{x}=\bm{P}_K\bm{E}_K^{\frac{1}{2}}\cvec{s}_K$ can be approximated  Gaussian  due to \ac{CLT} \cite{PaRV89}. By using Bussgang decomposition in \eqref{eq:bussgang_decom} and the fact that  $\op{E}[\cvec{x}\cvec{d}^{\op{H}}]\approx\bm{0}$ with  approximately Gaussian $\cvec{x}$, we write the following approximation
\begin{equation}
\begin{aligned}
&\op{tr}(\bm{C}_{\cvec{t}\cvec{t}})=\op{tr}(\op{E}[(\bm{B}\cvec{x}+\cvec{d})(\bm{B}\cvec{x}+\cvec{d})^{\op{H}}])\\
&\approx \xi_Q^2\op{tr}(\text{diag}(\bm{C}_{\cvec{x}\cvec{x}})^{-\frac{1}{2}}\bm{C}_{\cvec{x}\cvec{x}}\text{diag}(\bm{C}_{\cvec{x}\cvec{x}})^{-\frac{1}{2}})+\op{tr}(\bm{C}_{\cvec{d}\cvec{d}}).
\end{aligned}
\end{equation} 
Diagonals of $\bm{C}_{\cvec{t}\cvec{t}}$  and  $\text{diag}(\bm{C}_{\cvec{x}\cvec{x}})^{-\frac{1}{2}}\bm{C}_{\cvec{x}\cvec{x}}\text{diag}(\bm{C}_{\cvec{x}\cvec{x}})^{-\frac{1}{2}}$ consist of ones. As a result, power of the distortion $\cvec{d}$ reads as 
\begin{equation}
\op{tr}(\bm{C}_{\cvec{d}\cvec{d}})\approx(1-\xi_Q^2)N.
\end{equation}
Eventually, quantization distortion power  is computed as
\begin{equation}\label{eq:dist_pow}
P^{(\text{q})}_k\approx (1-\xi_Q^2)\gamma\sigma_k
\end{equation}
\subsubsection{Approximate SQINR }
At this step, we again resort to the asymptotic analysis assumption that  $K$ is large, which helps us to characterize the precoded signal $\cvec{x}$ as approximately Gaussian due to \ac{CLT}. In this case, the quantization distortion $\cvec{d}$ and $\cvec{s}$ can be considered  uncorrelated so that we can combine \eqref{eq:desired_fin}, \eqref{eq:P_mui} and  \eqref{eq:dist_pow} to approximate the SQINR at the $k$th user  as follows 	
\begin{equation}\label{eq:SQINR}
\text{SQINR}_k(w_k,u)\approx\frac{(\frac{1}{\beta}-u^2)}{\tau (1-u)^2+c_k} w_k,\; \text{for} \; k=1,2 \ldots K,  
\end{equation}
where we introduced and $\tau=\frac{K-1}{K}$ and 
\begin{equation}\label{eq:c_k}
c_k=\frac{1}{\xi_Q^2}(1+\frac{1}{\gamma \sigma_k})-1.
\end{equation}
An important fact that we use later is  $c_k>0$, since for \ac{CE} quantized systems $\xi_Q<1$. Furthermore, the approximate  SQINR in \eqref{eq:SQINR} is not changed with  positive scaling of $\bm{E}_K$ since  positive scaling of $\bm{E}_K$ does not change $(w_1,w_2,\ldots w_K)$ in \eqref{eq:w_k_def}. This is  due to \ac{CE} quantization being invariant to scaling with a positive term  such that  $\mathcal{Q}(\cvec{x})=\mathcal{Q}(v\cvec{x})$.

Equation \eqref{eq:SQINR} can be used to approximate the systems with infinite resolution quantization by setting $\xi_Q$ to 1. Such an approximation is  equivalent to approximation of \ac{SINR} in \cite{MuZaEv13}. Yet, due to especially making use of \eqref{eq:10} and introduction of $u$,  \eqref{eq:SQINR} is more compact  than its counterpart   \cite[Eq. (4)]{MuZaEv13}. Also, equation \eqref{eq:SQINR} is more intuitive  as it shows how regularization affects powers of the desired signal and \ac{MUI}, whereas such an observation cannot be made in   \cite[Eq. (4)]{MuZaEv13}.

\section{Approximately Equivalent System}\label{sec:approx_sys}
Equation  \eqref{eq:SQINR} shows the approximate \acp{SQINR}, when \ac{RZF} precoder is designed for $K$ users. The exact achievable rates of users cannot be  computed  via \eqref{eq:SQINR}, since  distribution of the quantization distortion $\cvec{d}$ is not available in an analytic form. Instead, a lower bound to the achievable rates of the system in \eqref{eq:SQINR} is defined by assuming a mismatched decoder at the receiver \cite{ArLoVo06}. We here consider that  receivers decode based on the Gaussian auxiliary channel, i.e., the decoding is done by assuming the quantization distortion is Gaussian distributed. This system with mismatched decoding is equivalent to the following system of $K$ parallel SISO channels   
	\begin{subequations}
	\label{eq:lower_bound}
	\begin{align}
	\label{eq:lower_bound_m_1}
	&y_k=\sqrt{\lambda_k(u)w_k} s_k+\eta_k,\\\label{eq:lower_bound_m_2}
	&\lambda_k(u)=\frac{(\frac{1}{\beta}-u^2)}{\tau (1-u)^2+c_k},\\\label{eq:lower_bound_m_3}
	&w_k=\frac{K \frac{e_k}{\sigma_k}}{\op{tr}(\bm{E}_K \bm{\Sigma}_K^{-1})} \;\text{for}\; k=1,2,\ldots K,
	\end{align}
	\end{subequations} 
where $s_k$ and $\eta_k$ are the same input and the noise signals as in \eqref{eq:io_decom}, and $\lambda_k(u)$ and $w_k$ are the channel gain and the transmit power of the $k$th user, respectively. Note that \eqref{eq:lower_bound_m_3} imposes a transmit power constraint $\sum_{k=1}^{K}w_k=K$ and the receive \ac{SNR} of the $k$th user  reads as $\text{SNR}_k(u,w_k)=\lambda_k(u)w_k$, which is equal to approximate SQINR in \eqref{eq:SQINR}. Furthermore,  the channel gains are in nonincreasing order such that $\lambda_1(u)\geq \lambda_2(u)\ldots \geq\lambda_K(u)$, because of large-scale fading coefficients being in nonincreasing order. For the sake of brevity, we introduce the transmit power vector $\cvec{w}={ \begin{bmatrix}
w_1, w_2, \ldots w_K \end{bmatrix}}^{\op{T}}$.

The $K$ parallel SISO channels system in \eqref{eq:lower_bound} is approximately equivalent to the quantized  MIMO downlink system in \eqref{eq:io_decom}. There is  a one-to-one relation between the original and the alternative regularization parameters $\alpha$ and $u$, which is given by \eqref{eq:rho}. A relation between $\cvec{w}$ and $\bm{E}_K$, which is helpful to transform one into the other,  is also present in \eqref{eq:lower_bound_m_3}. As a result, sum rate maximization of system in Fig.~\ref{fig:downlink_quan} can approximately  be done via the parallel SISO channels system in \eqref{eq:lower_bound}. To this aim, we formulate and solve the sum rate maximization problem for system in \eqref{eq:lower_bound}. 

\subsection{Sum Rate Maximization of the Approximate System}
The original sum rate maximization problem \eqref{eq:original_prob_with_K}  can be transformed to sum rate maximization for \eqref{eq:lower_bound} as follows
\begin{equation}\label{eq:original_prob_2}
\underset{\substack{u\in[0,1], \\ K \in \{1,2,\ldots M\},\\ \cvec{w}\geq \cvec{0}_K}}{\textrm{argmax}} \sum_{k=1}^{K} I(\lambda_k(u,K)w_k) \; \text{s.t} \; \sum_{k=1}^K w_k=K, 
\end{equation}
where $I(\lambda_k(u,K)w_k)=I(s_k;\sqrt{\lambda_k(u,K)w_k}s_k+\eta_k)$ is the mutual information between the input signal $s_k$ and received signal $y_k$. Notice that we changed the notation from  $\lambda_k(u)$
to $\lambda_k(u,K)$ since we now consider $K$ also as an optimization variable. Essentially, we now see the channel gain $\lambda_k(u,K)$ dependent on $K$ via $\beta(K)=\frac{K}{N}$ and $\tau(K)=\frac{K-1}{K}$. 

A well-known result from \cite{GuShVe05} states that the mutual information between the $k$th input and output $I_k(\lambda_k(u,K)w_k)$ is a concave function of $\lambda_k(u,K)w_k$. However, as $\lambda_k(u,K)$ is not concave in $u$, the objective function of \eqref{eq:original_prob_2} is not jointly concave in $\cvec{w}$, $u$ and $K$. Typically a nonconvex optimization  with $M+2$ variables such as  \eqref{eq:original_prob_2} is impractical  due to its  high computational complexity even for moderate values of $M$.  Yet, as we use the alternative regularization parameter $u$ in formulating \eqref{eq:original_prob_2} and the optimal power allocation for fixed $(u,K)$ can be computed, we can  provide a method that obtains the global maximum with an feasible computational complexity. Let us first discuss the power allocation problem with fixed $u$ and $K$, which provides a basis for the algorithms we devise to solve \eqref{eq:original_prob_2}.

\subsection{Power Allocation with Fixed $u$ and $K$}\label{sec:opt_pow_fix_u}
For a parallel SISO  system with fixed $u$ and $K$, the sum rate maximization is reduced to  the following convex problem 
\begin{equation} \label{eq:fixed_u}
\cvec{w}^{\star}(u,K)=\underset{\cvec{w}\geq \cvec{0}_K}{\textrm{argmax}} \sum_{k=1}^{K} I_k(\lambda_k(u,K) w_k) \; \text{s.t} \; \sum_{k=1}^K w_k=K.
\end{equation}
The optimal solution of  problem \eqref{eq:fixed_u} is given in   \cite{LoTuVe05} as 
\begin{equation}\label{eq:solution_1}
w_k^{\star}(u,K)=\frac{1}{\lambda_k(u,K)} \text{MMSE}^{-1}(\min(1,\frac{\mu}{\lambda_k(u,K)}))\; \forall k,
\end{equation}
with $\mu$ satisfying
\begin{equation}\label{eq:solution_2}
\sum_{\substack{k=1,\\\lambda_k(u,K)>\mu}}^{K} \frac{1}{\lambda_k(u,K)} \text{MMSE}^{-1}(\frac{\mu}{\lambda_k(u,K)})=K.
\end{equation}
The function $\text{MMSE}^{-1}(\bullet)$, with the domain in the interval $[0,1]$, is the inverse function of $\text{MMSE}(\text{SNR})$, which is the minimum mean square error achieved by the conditional mean estimator in the channel 
\begin{equation}
y_k=\sqrt{\text{SNR}}s_k+\eta_k,
\end{equation}
where $\eta_k\sim\mathcal{CN}(0,1)$ and $s_k$ has the same distribution as in \eqref{eq:lower_bound_m_1}. The value of $\text{MMSE}(\text{SNR})$ can be computed offline by integration over the complex field and it depends on the distribution of $s_k$. For the details on  computing  $\text{MMSE}(\text{SNR})$, the reader is referred to \cite{LoTuVe05}.  A look-up table can then be generated to implement the function $\text{MMSE}^{-1}(\bullet)$  by using the pre-computed values of $\text{MMSE}(\bullet)$. The value of $\mu$ satisfying \eqref{eq:solution_2} can be found by the bisection search method and using the generated look-up table. This procedure of obtaining solutions for \eqref{eq:solution_1} and \eqref{eq:solution_2} is referred as mercury/waterfilling \cite{LoTuVe05}.

\section{Sum Rate Maximization Algorithms}\label{sec:opt_alg}
We  here present the branch and bound method and the alternating optimization to solve  \eqref{eq:original_prob_2}. The former has significantly higher computational complexity than the latter, yet it  obtains the global maximum, whereas the latter algorithm converges to a local maximum. 
\subsection{Global Optimization via the Branch and Bound Method}\label{sec:bb}
As the sum rate maximizing power allocation vector for fixed $u$ and $K$ is given as  $\cvec{w}^\star(u,K)$,  problem  \eqref{eq:original_prob_2} can be rewritten as 
\begin{equation}\label{eq:prob_again}
\underset{(u,K)\in\mathcal{S}_0}{\max} \sum_{k=1}^K I(\lambda_k(u,K)w_k^\star(u,K)), 
\end{equation}
with $\mathcal{S}_0=\{(u,K)|u\in [0,1], K\in\{1,2,\ldots M\}\}$ such that the optimization variables are reduced to $u$ and $K$.

To solve \eqref{eq:prob_again}, we modify the \ac{MMP} framework from  \cite{MaHeJoUt20}, which enables an efficient application of the branch and bound algorithm \cite[Sec. 6.2]{Tuy98}. 
The branch and bound method consists of a search for the global optimum by systematically dividing  an initial constraint set into disjoint subsets and checking which subsets potentially contain the optimum via defined upper and lower bounds. In our case, the initial set which contains the global maximum is given as  $\mathcal{S}_0$,

For a given subset  $\mathcal{S}=\{(u,K)| \,u\in[u_\text{L},u_\text{U}],\,K\in\{K_{\text{L}},\ldots K_\text{U}\}\}$, the maximum sum rate can be formulated as  
\begin{equation}\label{eq:sum_rate_box}
R(\mathcal{S})=\underset{(u,K)\in\mathcal{S}}{\max} \underset{\cvec{w}\geq\cvec{0}}{\max} \sum_{k=1}^K I(\lambda_k(u,K)w_k) \;\text{for}\; \sum_{k=1}^{K}w_k=K.
\end{equation} 
An upper bound to \eqref{eq:sum_rate_box} can be defined as 
\begin{equation}\label{eq:upper_bound_box}
U(\mathcal{S})=\underset{\substack{\cvec{w}\geq \cvec{0}}}{\max} \sum_{k=1}^{K_\text{U}} I(\Lambda_k(u_\text{L},u_\text{U},K_\text{L})w_k) \; \text{s.t.} \; \sum_{k=1}^{K_\text{U}}w_k=K_\text{U},
\end{equation}
where we introduce the enhanced channel gain for the $k$th user as
\begin{equation}\label{eq:channel_enhance}
\Lambda_k(u_\text{L},u_\text{U},K_\text{L})=\frac{(\frac{1}{\beta(K_\text{L})}-u_\text{L}^2)}{\tau(K_\text{L}) (1-u_\text{U})^2+c_k}
\end{equation}
for $k=1,\ldots, K_{\text{U}}$. When introducing $\Lambda_k(u_\text{L},u_\text{U},K_\text{L})$, we apply the \ac{MM} formulation to $\lambda_k(u,K)$ such that we replace $u$ in the numerator and denominator of \eqref{eq:lower_bound_m_2} by  $u_\text{L}$  and $u_\text{U}$, respectively  and $K$ by $K_\text{L}$. As  the numerator and denominator in \eqref{eq:lower_bound_m_2}  are decreasing in $u$ and $\lambda_k(u,K)$ is decreasing in $K$, it holds that
\begin{equation}\label{eq:strong_channel}
\Lambda_k(u_\text{L},u_\text{U},K_\text{L})\geq\underset{(u,K)\in\mathcal{S}}{\textrm{max}}\lambda_k(u,K) \;\text{for} \; k=1,\ldots, K_\text{U}.
\end{equation}
To see the whole derivation of upper bound \eqref{eq:upper_bound_box}, the reader is referred to Appendix \ref{sec:upp_bound}.  

A lower bound to $R(\mathcal{S})$ can be generated by using any element $(u,K)\in\mathcal{S}$ and the corresponding optimal $\cvec{w}^{\star}(u,K)$. Here, we select  $(u_\text{U},K_\text{L})$ to generate the lower bound as follows
\begin{equation}\label{eq:lower_bound_box}
L(\mathcal{S})=\sum_{k=1}^{K_\text{L}} I(\lambda_k(u_\text{U},K_\text{L})w_k^{\star}(u_{\text{U}},K_\text{L})).
\end{equation}
With the upper and  lower bounds  defined in \eqref{eq:upper_bound_box} and \eqref{eq:lower_bound_box}, we apply two tests on  a subset to see if it potentially contains the solution of  \eqref{eq:prob_again}. The first one is to check if a subset's  upper bound is smaller than a lower bound of any other subset.  Such a subset is guaranteed not to contain the solution of \eqref{eq:prob_again} and it can be discarded from the search.   
The second test is to check if the optimal power allocation vector in \eqref{eq:upper_bound_box}, which we denote as $\cvec{w}_\text{U}(\mathcal{S})$, has less nonzero entries than $K_\text{L}$.  In Appendix \ref{sec:app},
\Cref{theorem:U_S} states that such subsets do not contain  the solution of \eqref{eq:prob_again}.

Note that  generating the upper and lower bounds with given  $\Lambda_k(u_\text{L},u_\text{U},K_\text{L})$, $\cvec{w}_\text{U}(\mathcal{S}),\lambda_k(u_\text{U}$, $K_\text{L})$, $w_k^{\star}(u_{\text{U}},K_\text{L})$ requires to compute  mutual information $I(\text{SNR})$ for a given SNR value. This can be again implemented  by offline numerical integration and use of the look-up tables. 

\begin{algorithm}[!]
\caption{Branch and Bound for \eqref{eq:prob_again}}
\begin{algorithmic}[1]\label{alg:bb}
\State Initialize union set all subsets  as $\mathbb{S}=\{\mathcal{S}_0\}$. 
\State Compute $U(\mathcal{S})$ for each subset in $\mathbb{S}$ by  \eqref{eq:upper_bound_box}.
\State Check $\cvec{w}_{\text{U}}(\mathcal{S})$  for each subset in $\mathbb{S}$. Remove  every subset with  $\cvec{w}_\text{U}(\mathcal{S})$ containing   less nonzero elements than its $K_\text{L}$.
\State  Compute  $L(\mathcal{S})$ for each subset in  $\mathbb{S}$ by \eqref{eq:lower_bound_box}
\State  Remove  subsets with $U(\mathcal{S})<\max_{\mathcal{S}\in\mathbb{S}}L(\mathcal{S})$.
\State Find the subset $\hat{\mathcal{S}}=\textrm{argmax}_{\mathcal{S}\in\mathbb{S}} U(\mathcal{S})$. Replace $\mathbb{S}$ by $\{\mathbb{S}\setminus\hat{\mathcal{S}}\}\cup\{\mathcal{S}_1,\mathcal{S}_2\}$ by using \eqref{eq:subset_div1} and \eqref{eq:subset_div2}.
\State Repeat  Steps 2-6  until $\max_{\mathcal{S}\in\mathbb{S}}U(\mathcal{S})-\max_{\mathcal{S}\in\mathbb{S}}L(\mathcal{S})\leq\epsilon$
\State  Return $u_{\text{U}}, K_\text{L}$ and $\cvec{w}^\star(u_\text{U},K_\text{L})$ of subset $\textrm{argmax}_{\mathcal{S}\in\mathbb{S}}L(\mathcal{S})$
\end{algorithmic}
\end{algorithm}

Application of the branch and bound method to solve \eqref{eq:prob_again} is described in \Cref{alg:bb}. First, the set $\mathbb{S}$, which is the union of all disjoint subsets that potentially contain the maximizer of \eqref{eq:prob_again}, is initialized as $\mathbb{S}=\{\mathcal{S}_0\}$. Then, the upper bounds to the sum rates of all disjoint subsets in $\mathbb{S}$ are computed. All subsets that have $\cvec{w}_{\text{U}}(\mathcal{S})$ with less nonzero elements than $K_\text{L}$ are removed from further search at Step 3. Then, the lower bounds to the maximum sum rates for the remaining subsets in $\mathbb{S}$ are computed. Subsets with a smaller upper bound than the maximum lower bound of all subsets are removed from  $\mathbb{S}$ at Step 5. At Step 6, the subset with the highest upper bound $\hat{\mathcal{S}}=\{(u,K)| \,u\in[\hat{u}_\text{L},\hat{u}_\text{U}],\,K\in\{\hat{K}_{\text{L}},\ldots\hat{K}_\text{U}\}\}$ is replaced by two disjoint subsets $\mathcal{S}_1$ and $\mathcal{S}_2$ which are determined in two following ways. If $\hat{u}_\text{U}-\hat{u}_\text{L}>\hat{K}_\text{U}-\hat{K}_\text{L}$, then it is set
\begin{equation}\label{eq:subset_div1}
\begin{aligned}
&\mathcal{S}_1=\{(u,K)| \,u\in[\hat{u}_\text{L},\frac{\hat{u}_\text{L}+\hat{u}_\text{U}}{2}],\,K\in\{\hat{K}_{\text{L}},\ldots\hat{K}_\text{U}\}\}\\
&\mathcal{S}_2=\{(u,K)| \,u\in[\frac{\hat{u}_\text{L}+\hat{u}_\text{U}}{2},\hat{u}_{\text{U}}],\,K\in\{\hat{K}_{\text{L}},\ldots\hat{K}_\text{U}\}\}
\end{aligned}
\end{equation} 
else if $\hat{u}_\text{U}-\hat{u}_\text{L}<\leq\hat{K}_\text{U}-\hat{K}_\text{L}$, it is set
\begin{equation}\label{eq:subset_div2}
\resizebox{\columnwidth}{!}{$\begin{aligned}
&\mathcal{S}_1=\{(u,K)| \,u\in[\hat{u}_\text{L},\hat{u}_{\text{U}}],\,K\in\{\hat{K}_{\text{L}},\ldots\floor{\frac{\hat{K}_{\text{L}}+\hat{K}_\text{U}}{2}}\}\}\\
&\mathcal{S}_2=\{(u,K)| u\in[\hat{u}_\text{L},\hat{u}_{\text{U}}],\, K\in\{\floor{\frac{\hat{K}_{\text{L}}+\hat{K}_\text{U}}{2}}+1,\ldots K_{\text{U}}\}\}.
\end{aligned}$}
\end{equation} 
Upper bounds of $\mathcal{S}_1$ and $\mathcal{S}_2$ are refined such that they are certainly less than or equal to the upper bound of $\hat{\mathcal{S}}$. As a result, repeating steps between 2 and 6 eliminates the subsets that are guaranteed not to contain the maximizer of \eqref{eq:prob_again} and the remaining subsets are divided into smaller subsets so that the gap between the remaining upper and lower bounds decrease. As the initial set $\mathcal{S}_0$ contains the global maximum, the convergence proof in \cite{MaHeJoUt20} suggests that it is guaranteed obtain $\epsilon$-optimal solution for \eqref{eq:prob_again} with \Cref{alg:bb}. 
Note that  the \ac{MM} formulation in \eqref{eq:channel_enhance}, which enables solving  \eqref{eq:prob_again} optimally, is not possible  without introduction of $u$.  
\subsection{Suboptimal Sum Rate Maximization via Alternating Optimization}\label{sec:alt_opt}
The second algorithm  we propose to solve  \eqref{eq:original_prob_2} is an alternating optimization process  where $\cvec{w}$ is updated by solving \eqref{eq:fixed_u} for a fixed $(u,K)$ and $u$ is updated by solving the sum rate maximization for fixed $\cvec{w}$ and $K$. To this aim, we  formulate the problem of finding the optimal $u$ with  fixed $K$ and $\cvec{w}$ as  
\begin{equation} \label{eq:fixed_w}
{u}^\star(\cvec{w},K)=\underset{\substack{u\in [0,1]}}{\textrm{argmax}} \sum_{k=1}^{K} I(\lambda_k(u,K)w_k).
\end{equation}
As $\lambda_k(u,K)$ is not concave in $u$, \eqref{eq:fixed_w} is not a convex problem. Yet, we can   obtain a local maximum for  \eqref{eq:fixed_w} by solving  
\begin{equation}\label{eq:deriv}
\sum_{k=1}^{K} \frac{\partial I(\text{SNR}_k(u, w_k))}{\partial u}\bigg|_{u=u^\star(\cvec{w},K)}=0.
\end{equation}
We switched notation to  $\text{SNR}_k(u,w_k)=\lambda_k(u,K)w_k$, for the sake of brevity. To obtain  $u^\star$, we first use the chain rule 
\begin{equation}\label{eq:der_decom}
\frac{\partial I(\text{SNR}_k(w_k,u))}{\partial u}= \frac{\partial I(\text{SNR}_k(w_k,u))}{\partial\text{SNR}_k(w_k,u)}\frac{\partial\text{SNR}_k(w_k,u)}{\partial u}.
\end{equation}
For the first term on the right in \eqref{eq:der_decom}, we know that \cite{GuShVe05} 
\begin{equation}\label{eq:mmse}
 \frac{\partial I(\text{SNR}_k(w_k,u))}{\partial\text{SNR}_k(w_k,u)}= \frac{\text{MMSE}(\text{SNR}_k(w_k,u))}{\mathrm{ln}2}.
\end{equation} 
The second term on the right hand side of \eqref{eq:der_decom} reads as 
\begin{equation}\label{eq:snr_drv}
\frac{\partial\text{SNR}_k(w_k,u)}{\partial u}=\frac{-2 c_k u+ 2 \tau(1-u)(\frac{1}{\beta}-u)}{(\tau(1-u)^2+c_k)^2}w_k.
\end{equation}
By using \eqref{eq:snr_drv} and \eqref{eq:mmse}, we execute a bisection search to find a  local maximum in $[0,1]$. We use the MMSE look-up table generated to compute the mercury/waterfilling solution of \eqref{eq:solution_2} to compute \eqref{eq:der_decom}. \Cref{theorem:existence}  implies a local maximum of \eqref{eq:fixed_w} is certainly obtained by bisection search for \eqref{eq:deriv}.

\begin{theorem}\label{theorem:existence}
 There exists at least one $u$ value  that satisfies \eqref{eq:deriv} in interval $[0,1]$.
\end{theorem}
\begin{proof}
In this proof, we denote $u^\star(\cvec{w},K)$ as $u^\star$ for the sake of brevity. First, let us consider the case with  $K=1$, where we compute that  
\begin{equation}\label{eq:snr_drv_M=1}
\frac{\partial\text{SNR}_1(w_1,u)}{\partial u}=-2\frac{u} {c_1}{w_1}.
\end{equation}
For $K=1$, it holds that
\begin{equation}\label{eq:sum_deriv}
\frac{\partial I(\text{SNR}_1(w_1,u))}{\partial u}\bigg|_{u=u^\star}=-2\frac{\text{MMSE}(\text{SNR}_1(w_1,u^{\star}))}{\mathrm{ln}2}\frac{u^{\star} w_1}{c_1},
\end{equation}
which is equal to 0 only if $u^{\star}=0$, since $c_1$ and $\text{MMSE}(\text{SNR}_1(w_1,u^\star))$ are strictly positive. Hence,  a local maximum in $[0,1]$ is obtained at  $u^{\star}=0$.

Now let us consider  $K\geq2$.
Observe that the denominator in \eqref{eq:snr_drv} is  positive and $w_k$ is nonnegative. Hence, y the sign of \eqref{eq:snr_drv} is simply determined by the sign of the convex quadratic expression $-2 c_k u+ 2 \tau(1-u)(\frac{1}{\beta}-u)$.  We compute that   
\begin{equation}\label{eq:u=0}
(-2 c_k u+ 2 \tau(1-u)(\frac{1}{\beta}-u))\bigg|_{u=0}=2 \frac{\tau}{\beta} > 0,
\end{equation}
which implies $\sum_{k=1}^{K} \frac{\partial I(\text{SNR}_k(w_k,u))}{\partial u}\bigg|_{u=0}>0$. For $u=1$, it holds that
\begin{equation}\label{eq:u=1}
{(-2 c_k u+ 2 \tau(1-u)(\frac{1}{\beta}-u))\bigg|}_{u=1}=-2c_k<0,
\end{equation}
which implies that $\sum_{k=1}^{K} \frac{\partial I(\text{SNR}_k(w_k,u))}{\partial u}\bigg|_{u=1}< 0$. Thus, there exists a sign change from plus to minus implying the existence of at least one local maximum of the sum rate in $u\in[0,1]$. 	
\end{proof}



\Cref{alg:alter} describes the alternating optimization we propose to solve \eqref{eq:original_prob_2}.  The algorithm is initialized with  the \ac{ZF} precoder designed for all $M$ users within the coverage area, i.e., $u^{(0)}=1$ and $K^{(0)}=M$.  After the initialization,  updates of $\cvec{w}$, $K$ and $u$ take place in order  until a while loop convergences. A single iteration of the while loop starts with the update of power allocation vector $\cvec{w}$ at Step 4 according to the most recent channel gains of active users, which are computed in Step 3 as  $(\lambda_1, \lambda_2,\ldots, \lambda_{K^{(i)}})$. At Step 5, the update of $K$ is done by checking if
$\cvec{w}^{(i+1)}$ has zeros. If there are no zeros, the update is as  $K^{(i+1)}= K^{(i)}$. On the other hand, if  $\cvec{w}^{(i+1)}$ has zeros, then designing the \ac{RZF} precoder with  $K^{(i)}$ users is  suboptimal.  \Cref{lemma:obs_2} in  Appendix~\ref{sec:app} implies that, higher sum rate is achieved by  removing zeros from $\cvec{w}^{(i+1)}$,  setting  $K^{(i+1)}$ as the new length of  $\cvec{w}^{(i+1)}$ and scaling $\cvec{w}^{(i+1)}$. 
At Step 6, $u^{(i+1)}$ is updated  by solving \eqref{eq:fixed_w} for fixed $ K^{(i+1)}$ and  $\cvec{w}^{(i+1)}$. The process between Steps 3-7 is repeated until convergence for which we select the criteria as $|u^{(i+1)}-u^{(i)}|<\epsilon$.
\begin{algorithm}[!]	
\caption{Alternating Optimization to Solve \eqref{eq:original_prob_2}}
\begin{algorithmic}[1]\label{alg:alter}
\State Initialize $i=0$, $u^{(0)}=1$, $K^{(0)}=M$, $\beta^{(0)}=\frac{K^{(0)}}{N}$, $\tau^{(0)}=\frac{K^{(0)}-1}{K^{(0)}}$.
\State \textbf{while} (termination criteria is not met)
\State Compute $(\lambda_1,\lambda_2,\ldots,\lambda_{K^{(i)}})$ with $u^{(i)}$, $K^{(i)}$, $\beta^{(i)}$, $\tau^{(i)}$ by using \eqref{eq:lower_bound_m_2}.
\State Solve \eqref{eq:fixed_u} with computed $(\lambda_1,\lambda_2,\ldots,\lambda_{K^{(i)}})$ by the mercury/waterfilling to obtain $\cvec{w}^{(i+1)}$ 
\State Check $\cvec{w}^{(i+1)}$ and remove if there are any zeros.   Update $K^{(i+1)}$ as number of elements of  $\cvec{w}^{(i+1)}$. Update $\beta^{(i+1)}=\frac{K^{(i+1)}}{N}$, $\tau^{(i+1)}=\frac{K^{(i+1)}-1}{K^{(i+1)}}$. Update $\cvec{w}^{(i+1)}= \frac{K^{(i+1)}}{K^{(i)}}\cvec{w}^{(i+1)}$. 
\State Solve \eqref{eq:fixed_w} with fixed $\cvec{w}^{(i+1)}$, $K^{(i+1)}$  to obtain $u^{(i+1)}$.
\State Update $i=i+1$
\State \textbf{end}
\State Return $u^{(i)}$,  $K^{(i)}$ and $\cvec{w}^{(i)}$
\end{algorithmic}
\end{algorithm}
\subsection{Mapping Back to the Original Parameters}\label{sec:map_back}
Let us denote  the optimized  parameters returned by \Cref{alg:bb} and \Cref{alg:alter} as  $u^\star$, $\cvec{w}^\star$ and $K^\star$. The  obtained $(u^\star$, $\cvec{w}^\star)$   must be mapped back to original parameters $(\alpha,\bm{E}_K)$ to precode $\cvec{x}$ in the actual system in Fig.~\ref{fig:downlink_quan}. The regularization parameter is mapped back as  $\alpha^\star=N(\frac{1}{u ^\star}-\beta(K^\star))(1-u^\star)$ by using the one-to-one relation in \eqref{eq:rho}. On the other hand, the relation between $\cvec{w}$ and $\bm{E}_K$ is not one-to-one and  any positive scaling of of matrix $\text{diag}({\cvec{w}^\star})\bm{\Sigma}_{K^{\star}}$ can be used as $\bm{E}_K$, we simply select  $\bm{E}_{K}=\text{diag}(\cvec{w}^\star)\bm{\Sigma}_{K^\star}$.  Eventually, the \ac{RZF} precoder is formed as $\bm{P}_K= \bm{H}_{K^\star}^{\op{H}}(\bm{H}_{K^\star}\bm{H}^{\op{H}}_{K^\star}+\alpha\bm{\Sigma}_{K^\star})^{-1} \in\mathbb{C}^{N\times K^\star}$ and precoding is done as $ \cvec{x}=\bm{P}_K{\bm{E}^{\star}_{K}}^{\frac{1}{2}}\cvec{s}_{K^\star}$.


\section{Analysis of \Cref{alg:alter}}\label{sec:alter_alg}
A rigorous analysis for overall computational complexity of \Cref{alg:bb} is very difficult, since the number of iterations it requires for convergence depends heavily on the channel realization (especially Step 5).  Yet, we empirically recognize  that its  computational complexity is significantly higher than \Cref{alg:alter}. On the other hand, although \Cref{alg:alter} obtains a local maximum for \eqref{eq:original_prob_2}, numerical results in \Cref{{sec:comp}} demonstrate that its performance is practically identical to \Cref{alg:bb}, which is proven to obtain the global maximum. We thus select \Cref{alg:alter} as the main method to solve \eqref{eq:original_prob_2}, whereas  \Cref{alg:bb} serves as the benchmark assuring that \eqref{eq:original_prob_2} is solved optimally. Accordingly, we  provide the high transmit power and computational complexity analysis only for \Cref{alg:alter}.

\subsection{High Transmit Power Regime Analysis}\label{sec:high_tr}
When \Cref{alg:alter} is initialized at high transmit power regime $\gamma\to\infty$ with $K^{(0)}=M$, the channel gains of the approximate system in  \eqref{eq:lower_bound_m_2} are computed as   $\lambda_1(u)\approx\lambda_2(u)\approx \lambda_{M}(u)=\lambda(u)$, since  $c_1\approx c_2\approx\ldots c_{M}\approx c=\frac{1}{\xi_Q^2}-1$.  In this case, the power allocation at Step 4  yields as   $\cvec{w}\approx \cvec{1}_M$ in every  iteration. Due to \eqref{eq:lower_bound_m_3}, this outcome implies that  \eqref{eq:high_transmit_power} holds with $K=M$.  In conclusion, the assumption that \eqref{eq:high_transmit_power} holds and $K$ is large, on which the approximate parallel SISO system in \eqref{eq:lower_bound} is based on, is accurate for systems with large $M$ at $\gamma\to\infty$, when \Cref{alg:alter} is employed. This is why we refer to it as high transmit power assumption. Also for \Cref{alg:bb}, which is the other algorithm that is based on the assumption that \eqref{eq:high_transmit_power} holds and $K$ is large, it is observed that every power allocation vector  $\cvec{w}^\star(u_\text{U},K_\text{L})$ that is computed as a candidate for the global maximizer is approximately  uniform when $\gamma\to\infty$. Hence, the high transmit power assumption is valid also for \Cref{alg:bb} at $\gamma\to\infty$. Note that due to Step 3 of \Cref{alg:bb} and Step 5 of \Cref{alg:alter}, the resulting $\cvec{w}^\star$ is not allowed to contain zeros, which improves the accuracy of the assumption in  \eqref{eq:high_transmit_power}.

The following theorem identifies the relationship between total number  of users $M$  and the obtained regularization parameter $u$, when $\gamma\to\infty$. 
\begin{theorem}\label{theorem:beta_u}
Consider two separate \ac{CE} quantized systems where a transmitter with $N$ antennas serves $M_1$ and $M_2$  users at $\gamma\to\infty$. Assume $N\geq M_2>M_1\geq 2$. For $u_1$ and $u_2$, which are  values of $u$ obtained by  \Cref{alg:alter} for  systems with $M_1$  and $M_2$ antennas, respectively, it holds that  $u_1>u_2$.
\end{theorem}
\begin{proof}
Previously in this subsection, it is shown that the power allocation vector is computed as $\cvec{w}=\cvec{1}_M$, when \Cref{alg:alter} is employed at $\gamma\to\infty$.  The problem of obtaining optimal $u$ at Step 6 via solving \eqref{eq:fixed_w} with the fixed $\cvec{w}=\cvec{1}_M$ boils down to problem of maximizing common SNR of all users, i.e., to the following problem 
 \begin{equation}\label{eq:u_update}
u^{\star}=\underset{u\in [0,1]}{\textrm{argmax}}\;\frac{\frac{1}{\beta(M)}-u^2}{\tau(M)\,(1-u)^2+c}.
\end{equation} 
Recall that $\beta(M)=\frac{M}{N}$ and $\tau(M)=\frac{M-1}{M}$. Note that we can compute  $\frac{\partial\text{SNR}(u,M)}{\partial u}$ as a specific version of \eqref{eq:snr_drv} with $w_k=1$ and $c_k=c>0$. As it is shown in  \Cref{{sec:alt_opt}},  the sign of the convex quadratic expression $-2 c u+ 2 \tau(M)\,(1-u)(\frac{1}{\beta(M)}-u)$ is simply the sign of $\frac{\partial\text{SNR}(u,M)}{\partial u}$.  Equation \eqref{eq:u=1} shows that there are  two roots  for  $\frac{\partial \text{SNR}(u,M)}{\partial u}$  and   $u=1$  is located between these roots. Furthermore, equation \eqref{eq:u=0} guarantees that the root smaller than $1$ is located in interval $[0,1]$ and  
it is the only maximizer  in $[0,1]$. Here, we can conclude that  there exists a single local maximizer in $[0,1]$. For a system with given $M_1$, let us denote the maximizer as  $u_1$.  

The derivative of $\frac{\partial \text{SNR}(u,M)}{\partial u}$ w.r.t. $M$ leads to 
\begin{equation}\label{eq:second_diff}
\frac{\partial^2 \text{SNR}(u,M)}{\partial u \partial M}= \frac{\partial (\frac{-2 c u+ 2 \tau(M)(1-u)(\frac{1}{\beta(M)}-u)}{(\tau(M)(1-u)^2+c)^2})} {\partial M}<0.
\end{equation}
Inequality \eqref{eq:second_diff} can be confirmed by inspecting that the numerator is decreasing  and the denominator is increasing in $M$\footnote{In the denominator the only quantity that depends on $M$ is $\tau(M)=\frac{M-1}{M}$, which is increasing in $M$. The derivative of the numerator with respect to $M$ reads as $2(1-u)\frac{N}{M^3}(2-M)-\frac{u}{M^2}$ which negative for $M\geq 2$}. As a result, it is implied that $\frac{\partial \text{SNR}(u,M)}{\partial u}$ is strictly decreasing in $M$. We can write that 
\begin{equation}\label{eq:92}
\frac{\partial \text{SNR}(u,M_2)}{\partial u}\bigg|_{u=u_1}<\frac{\partial \text{SNR}(u,M_1)}{\partial u}\bigg|_{u=u_1}=0
\end{equation}
for  $M_2>M_1$.  We already know that for the system with of $M_2$, there exists the maximizer $u_2$ in  $[0,1]$ for which we can form the following inequality
\begin{equation}\label{eq:root_rule_2}
\frac{\partial \text{SNR}(u,M_2)}{\partial u}\bigg|_{u=u_1}<\frac{\partial \text{SNR}(u,M_2)}{\partial u}\bigg|_{u=u_2}=0.
\end{equation} 

Since $u_2$ is the only root of  $\frac{\partial I(\text{SNR}(u,M_2)}{\partial u}$  in interval $[0,1]$  and the sign change at $u=u_2$ is from plus to minus, it  holds that $u_1>u_2$. This concludes the proof of \Cref{theorem:beta_u}. 
\end{proof}
\Cref{theorem:beta_u} suggests that at high transmit power regime, \ac{ZF}  precoding becomes more optimal as number of users $M$ decreases down to 2. We can extend \Cref{theorem:beta_u} to the infinite resolution scenario, where $c=0$. In this case, $u=1$ is the left root of the convex quadratic expression  $2 \tau(M)\,(1-u)(\frac{1}{\beta(M)}-u)$ that determines the sign of $\frac{\partial\text{SNR}(u,M)}{\partial u}$ (The other root satisfies $\frac{1}{\beta}$ is greater than 1, since $\beta\leq 1$). Thus, in infinite resolution systems with $\beta\leq 1$ the maximum in interval $[0,1]$ is obtained at $u=1$, which is the solution that is returned by \Cref{alg:alter} at $\gamma\to\infty$. This means that the difference between $u$ values obtained by  \Cref{alg:alter} in quantized and infinite resolution systems increase as number of users increase. Moreover, note that our outcome is in parallel with the finding in \cite{XuXuGo19}, which considers quantized systems with no power allocation.  Therein the optimal regularization parameter  $\rho$ is revealed to be proportional to the user load $\beta$ at high transmit power regime.
 \subsection{Computational Complexity of RZF precoding with \Cref{alg:alter}}\label{sec:complexity}
Computational burden of linear precoding consists of obtaining the precoding matrix once at every coherence time and linear transformation of $\cvec{s}$ to $\cvec{x}$ at every transmission. In our case, we obtain the solution for the approximate sum rate maximization in  \eqref{eq:original_prob_2} as $u^\star$, $\cvec{w}^\star$ and $K^\star$ by \Cref{alg:alter}, then we map them to $\bm{E}_K$ and $\bm{P}_K$ as it is described in  \Cref{sec:map_back}. 
Before employing \Cref{alg:alter},  we compute all large-scale fading coefficients as $\sigma_m=\frac{{\|\cvec{h}_m\|}^2}{N}$ and  we compute $(c_1,c_2, \ldots c_M)$ by \eqref{eq:c_k}. 
These operation cost $2MN$  and $3M$ \acp{FLOP}, respectively. 

The next part that requires \acp{FLOP} is the while loop of \Cref{alg:alter}.
Let us denote the number of iterations that the while loop takes to converge as $I$. Mercury/waterfilling and bisection search are two main operations that are executed at every iteration  to solve  \eqref{eq:fixed_u}  and  \eqref{eq:fixed_w}, respectively. In the first iteration of the while loop,  mercury/waterfilling algorithm is  applied with $M$ active users. In most cases, the number of active users does not change after the first iteration so that for $I-1$ iterations mercury/waterfilling algorithm is  applied with $K^{(I)}$ active users, where $K^{(I)}$ is the number of active users that \Cref{alg:alter} obtains as solution. In application  of mercury/waterfilling, \eqref{eq:solution_2} is solved by employing the bisection search and a look-up table. Here, we neglect the computational complexity of accessing the memory where the look-up tables are stored. Each iteration of bisection search to solve \eqref{eq:solution_2} with $K$ users takes $2K$ multiplication and $K-1$ additions. Hence, an iteration of bisection search to solve \eqref{eq:solution_2} with $K$ users costs $\mathcal{O}(3K)$.
We empirically observe that equation \eqref{eq:solution_2} is satisfied by error of $10^{-6}$ at approximately $i_\text{mercury}\approx 22$nd iteration of the bisection search, regardless of value of $K$. We calculate the total computational complexity of solving \eqref{eq:solution_2} in \Cref{alg:alter} as   $i_{\text{mercury}}(3M+3K^{(I)}(I-1))$ \acp{FLOP}.

We solve \eqref{eq:fixed_w} by finding  the regularization parameter $u$  satisfying \eqref{eq:deriv} with bisection search  in interval $[0,1]$. The solution of \eqref{eq:deriv} can be found by the error of $10^{-6}$ in $i_\text{bisection}=20$ iterations.  Each iteration of bisection search consists of computing \eqref{eq:der_decom} by using the MMSE look-up tables that we also used to solve  \eqref{eq:solution_2}. For $K$ active users,  computation of the numerator and the  denominator in \eqref{eq:snr_drv} costs  $2K$ operations each. Computing the fraction in   \eqref{eq:snr_drv} and multiplication by $w_m$s cost $K$ FLOP each. Multiplication with MMSE values from \eqref{eq:mmse} costs another  $K$ operations and $K-1$ additions more are required to compute \eqref{eq:deriv}. As a result, one iteration of bisection search to solve  \eqref{eq:deriv} costs approximately $8K$ FLOPs. The total complexity of finding the $u$ satisfying \eqref{eq:deriv} is $i_\text{bisection}(8M+(I-1)K^{(I)})$.

Computing  matrix $\bm{P}_K=\bm{H}_{K}^{\op{H}}(\bm{H}_{K}\bm{H}^{\op{H}}_{K}+\alpha\bm{\Sigma}_{K})^{-1}$ is the operation that costs the highest number of FLOPS after completion of the while loop. Computing the term $\bm{H}_K\bm{H}_K^{\op{H}}$ consists of approximately $K^2N$ operations (taking into account that $\bm{H}_K\bm{H}_K^{\op{H}}$ is a Hermitian matrix).  Computing $(\bm{H}_{K}\bm{H}^{\op{H}}_{K}+\alpha\bm{\Sigma}_{K})$ hence costs $K^2N+2K$ FLOPs. If we assume that Cholesky decomposition, forward-substitution and back-substitution is employed to compute the inversion, then the inversion $(\bm{H}_{K}\bm{H}^{\op{H}}_{K}+\alpha\bm{\Sigma}_{K})^{-1}\bm{H}_{K}$  costs $\frac{K^3}{3}+2K^2N$ FLOPS. Also, scaling by the entries of $\bm{E}^{\frac{1}{2}}$ costs $KN$ FLOPs. As a result, the complexity of computing $\bm{P}_{K}\bm{E}_K^{\frac{1}{2}}$ is in total $\frac{K^3}{3}+3K^2N+KN+2K$ FLOPs. Computing $\bm{E}_K$ costs $K$ multiplications and $K$ square root operations. 

In most channel realizations, $I$ required for \Cref{alg:alter} to converge is 3. Thus, as number of active users $K$ and $N$ grow large, most of computational complexity is caused by the inversion operation and complexity due to computing the regularization parameter and the power factors becomes negligible. The complexity when $K\to\infty,M\to\infty$ and $N\to\infty$ is $\mathcal{O}(K^3+K^2N)$.

\section{Extension of Q-GPI-SEM to CE Quantization}\label{sec:ex_QA_GPI_SEM}
In \cite{AsXuNo23_2}, performance of \Cref{alg:alter} in  1-bit MIMO downlink is compared to the state-of-the-art Q-GPI-SEM from \cite{ChPaLe22}. Q-GPI-SEM cannot be applied directly in CE MIMO downlink with higher resolution, since it relies on the \ac{AQNM} which is valid if  the real and imaginary parts of the precoded signal are quantized separately. A nontrivial extension of Q-GPI-SEM is required. To this aim, we utilize the \ac{LCA}  model which has been first presented in  \cite{MeGhNo09} and enabled for use of CE quantization in \cite[Chapter 3]{HeTh}.

With few modifications, Fig. ~\ref{fig:downlink_quan} can be used to represent the system model with Q-GPI-SEM as well.  Unlike Fig. ~\ref{fig:downlink_quan}, Q-GPI-SEM precoding is done as $\cvec{x}=\bm{P}\cvec{s}$ such that the power allocation is not carried out via matrix $\bm{E}$ but it is included in precoding matrix $\bm{P}$. Also, a transmit power constraint is imposed on the precoded signal such that $\op{tr}(\bm{P}\bm{P}^{\op{H}})=\gamma$ and the input $\cvec{s}$ is assumed to be Gaussian. The steps after precoding remain the same as Fig. ~\ref{fig:downlink_quan}. Given these changes, let us first consider the following  quantization
\begin{equation}\label{eq:CE_quantization_opt}
\mathcal{Q}_{\text{LCA}}(x_n)=r_n\exp(\text{j}(\ceil{\frac{\angle{x_n}}{2\psi}}-\psi)), \forall n, 
\end{equation}  
where $r_n$ is the envelope of the quantized output and $\psi=\frac{\pi}{Q}$.
On the contrary to \eqref{eq:CE_quantization}, each element of $\cvec{x}$ is quantized with a different envelope $r_n$. In such quantization operation, one has a freedom to select $r_n$ as the value minimizing $\op{E}[|\mathcal{Q}_{\text{LCA}}(x_n)-x_n|^2]$. In case of  $x_n\sim\mathcal{CN}(0,\sigma_{x_n}^2)$, the $r_n$ value minimizing $\op{E}[|\mathcal{Q}_{\text{LCA}}(x_n)-x_n|^2]$ is given as $r_n=\xi_Q \sigma_{x_n}$ \cite{HeTh}.  As a result, when quantization in \eqref{eq:CE_quantization_opt} is applied to all elements of $\cvec{x}$ with corresponding minimizing $r_n$, the output can be decomposed by \ac{LCA} as follows \cite{HeTh}
\begin{equation}\label{eq:lca_init}
\mathcal{Q}_{\text{LCA}}(\cvec{x})= \xi_Q^2 \cvec{x}+\cvec{d}_{\text{LCA}},
\end{equation}
where $\cvec{x}$ and $\cvec{d}_{\text{LCA}}$ are uncorrelated.  The covariance of distortion $\cvec{d}_{\text{LCA}}$ is given as \cite{HeTh}
\begin{equation}\label{eq:dist_cov}
\bm{C}_{\cvec{d}_{\text{LCA}}\cvec{d}_{\text{LCA}}}\approx\xi_Q^2(1-\xi_Q^2)\text{diag}(\bm{P}\bm{P}^{\op{H}}). 
\end{equation} 
We can reformulate \eqref{eq:CE_quantization} in terms of $\mathcal{Q}_{\text{LCA}}(\cvec)$ such that 
\begin{equation}\label{eq:ce_quan_reform}
\mathcal{Q}(\cvec{x})=\frac{1}{\xi_Q}{\text{diag}(\bm{P}\bm{P}^{\op{H}})}^{-\frac{1}{2}}\mathcal{Q}_{\text{LCA}}(\cvec{x}).
\end{equation}
At this step, we again resort to the asymptotic approximation which revealed in \eqref{eq:15} that for RZF precoder $\text{diag}(\bm{C}_{\cvec{x}\cvec{x}})$ converges to a scaled identity matrix. We also here expect that for a reasonable choice of $\bm{P}$,  $\text{diag}(\bm{P}\bm{P}^{\op{H}})$ asymptotically converges to a scaled identity matrix. Due to the constraint $\op{tr}(\bm{P}\bm{P}^{\op{H}})=\gamma$, we approximate $\text{diag}(\bm{P}\bm{P}^{\op{H}})$ as $\frac{\gamma}{N}\eye_N$. By employing  \eqref{eq:ce_quan_reform} and \eqref{eq:lca_init}, we can approximate the received signal at the $m$th user
\begin{equation}\label{eq:y_decom}
\begin{aligned}
y_m\approx\xi_Q \cvec{h}^{\op{H}}_m\cvec{p}_m s_m+ \xi_Q\sum_{j\neq m} \cvec{h}_m^{\op{H}}\cvec{p}_j s_j+ \frac{1}{\xi_Q}\cvec{h}_m^{\op{H}}\cvec{d}_{\text{LCA}}+\eta_m.
\end{aligned}
\end{equation} 
The received quantization distortion power  is approximated as 
\begin{equation}\label{eq:dist_pow_aqnm}
\frac{1}{\xi_Q^2}\cvec{h}_m^{\op{H}}\bm{C}_{\cvec{d}_{\text{LCA}}\cvec{d}_{\text{LCA}}}\cvec{h}_m =(1-\xi_Q^2)\sum_{j=1}^M \cvec{p}_j^{\op{H}} \text{diag}(\cvec{h}_m\cvec{h}_m^{\op{H}})\cvec{p}_j,
\end{equation}
by employing \eqref{eq:dist_cov} and some algebraic manipulations.
As we assume Gaussian inputs, \eqref{eq:y_decom} and \eqref{eq:dist_pow_aqnm} can  be combined to approximate the rate of the $m$th user as 
\begin{equation}\label{eq:rate_exp}
R_m\approx\log_2(\frac{\bar{\cvec{p}}^{\op{H}}\bm{C}_m\bar{\cvec{p}}}{\bar{\cvec{p}}^{\op{H}}\bm{D}_m\bar{\cvec{p}}}), 
\end{equation}
where
\begin{equation}\label{eq:matrices}
\begin{aligned}
&\bm{C}_m=\text{blkdiag}(\bm{G}_m,\bm{G}_m\ldots\bm{G}_m)+\frac{1}{\gamma}\eye_{MN},\\
&\bm{G}_m=\xi_Q^2 \cvec{h}_m\cvec{h}_m^{\op{H}}+ (1-\xi_Q^2)\text{diag}(\cvec{h}_m\cvec{h}_m^{\op{H}}),\\
&\bm{D}_m=\bm{C}_m-\text{blkdiag}(\bm{0}_N,\ldots, \xi_Q^2 \cvec{h}_m\cvec{h}_m^{\op{H},}\ldots \bm{0}_N)
\end{aligned}
\end{equation}
and $\bar{\cvec{p}}$ is the unit norm vector that is obtained by  stacking and normalizing $\bm{P}$, i.e., $\bar{\cvec{p}}=\frac{\text{vec}(\bm{P})}{\sqrt{\gamma}}$. 
In \cite{ChPaLe22}, Q-GPI-SEM is offered as an algorithm that obtains the unit norm solution to the rate maximization problem where rate expressions are in form of \eqref{eq:rate_exp}. We can apply Q-GPI-SEM from \cite{ChPaLe22} with $\bm{C}_m,\bm{D}_m$  in \eqref{eq:matrices} for $Q$-level CE quantization  (see Section V in  \cite{ChPaLe22} for the details). The unit norm stacked vector obtained by Q-GPI-SEM is scaled by $\sqrt{\gamma}$  and reshaped as separate columns in  $\bm{P}\in\mathbb{C}^{N\times M}$ obeying $\op{tr}(\bm{P}\bm{P}^{\op{H}})=\gamma$.  

Q-GPI-SEM performs power iterations, where the operation with highest computational cost is inversion of a $N\times N$ matrix. At each power iteration, matrix inversion is performed $M$ times so that one power iteration's computational complexity is $\mathcal{O} (MN^3)$. For that reason, when the number of active users  $K<<N$, \Cref{alg:alter} has significantly lower computational complexity than Q-GPI-SEM. Furthermore, linear transformation of $\cvec{s}$ to $\cvec{x}$ costs $2MN$ FLOPs in Q-GPI-SEM , whereas  the precoding is done by $\cvec{x}=\bm{P}_k\bm{E}_K\cvec{s}_K$ in \Cref{alg:alter}, which costs $2KN$ FLOPs. If not all users are active , i.e., $K<M$, then we save computational complexity with \cref{alg:alter}, which gets larger if the bandwidth is increased.
\section{Numerical Results}\label{sec:num_res}
In this section, we provide numerical results for the proposed \ac{RZF} precoding techniques. We select empirical generalized mutual information (GMI) as the  performance metric illustrating  the user rates (for the details on how to compute empirical GMI the reader is referred to \cite{NeSt18}). Average rate over all $M$ users within the cell is plotted versus transmit SNR $\gamma_{\text{dB}}$. Curves for the average rates are obtained by averaging over 1000 realizations of the  Rayleigh fading channel model. At each channel realization, the large-scale fading coefficient of the $m$th user is determined from the path loss (PL) as $\sigma_{m}=\frac{1}{\text{PL}_m}$ and  the PL of the $m$th user is generated as 
\begin{equation}\label{eq:pl_vector_dB}
\text{PL}_{m} \text{ (in  dB)}=a+10 b\log_{10}(z_m)+{\zeta},
\end{equation}
where $z_m$ and  $\zeta\sim\mathcal{N}(0,\sigma_\zeta^2)$ are the distance from the transmitter  and the shadowing factor of the $m$th user, respectively.  Parameters $a,b,\sigma_\zeta^2$  are  given  as $61.4, 3.4, 9.7$, respectively, which correspond to the non-line-of-sight (NLOS) channel measurements at 28 GHz \cite{Samimi16}.   The single-antenna users are distributed  in a ring-like area with inner radius of 35 m and outer radius of 200 m with uniform probability.
Finally, we assume equiprobable distribution for the scenarios with finite input constellations.
\subsection{Comparison of \Cref{alg:bb} and \Cref{alg:alter}}\label{sec:comp}
In Fig.~\ref{fig:bbvsbisect}, we see the average rates with \ac{RZF} precoding  obtained by \Cref{alg:bb} and \Cref{alg:alter} in CE MIMO downlink with $N=64$ antennas. Solid and  dashed curves depict the average rates achieved by \Cref{alg:bb} in systems with $M=8$ and $M=32$ users, respectively. Colors of the curves are selected according to the combination input constellation and number of CE quantization levels $Q$. The diamond marks are obtained by applying \Cref{alg:alter} under the same settings as the solid or dashed curves they overlap with. The  performance gap between \Cref{alg:bb} and  \Cref{alg:alter} is negligible in all depicted settings and  we empirically confirm that \Cref{alg:bb} has significantly higher computational complexity \footnote{For example, in CE MIMO downlink with $M=8$, $N=64$, $Q=4$ and QPSK inputs, \Cref{alg:bb} takes on the average 8 times to converge in comparison to \Cref{alg:alter}.}. For that reason, quantization-aware \ac{RZF} (QA-RZF) curves are generated with \Cref{alg:alter} for the rest of the numerical results.   
\begin{figure}[!]
\centering
\input{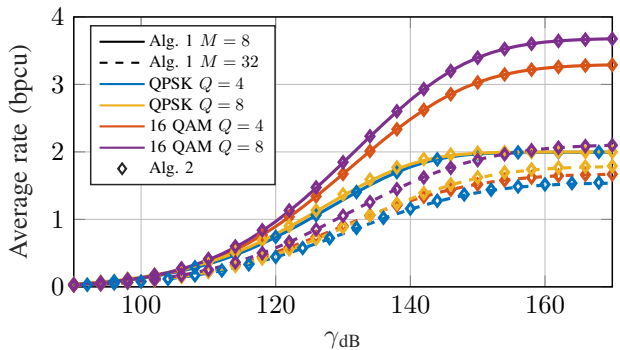}
\caption{Average rates of users in  CE quantized MIMO downlink with $N=64$. $M=8$ for solid curves, $M=32$ for dashed curves. Different colors represent different input signal constellation and quantization. The diamond marks are obtained by \Cref{alg:alter}.} 
\label{fig:bbvsbisect} 
\end{figure}
\subsection{Benefits of Quantization Awareness} 

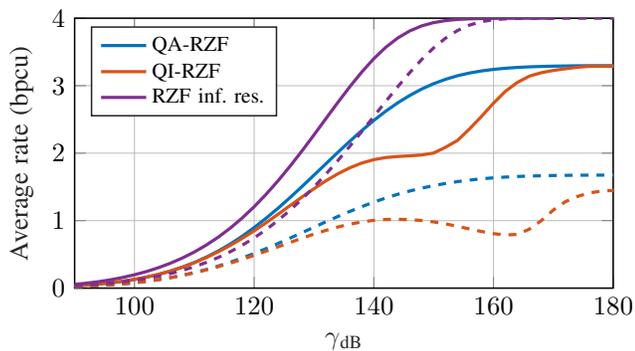
\begin{figure}[!]
%
%
\definecolor{mycolor1}{rgb}{0.00000,0.44700,0.74100}%
\definecolor{mycolor2}{rgb}{0.00000,0.44706,0.74118}%
\definecolor{mycolor3}{rgb}{0.85098,0.32549,0.09804}%
\definecolor{mycolor4}{rgb}{0.49400,0.18400,0.55600}%
\begin{tikzpicture}

\begin{axis}[%
width=0.951\fwidth,
height=0.477\fwidth,
at={(0\fwidth,0\fwidth)},
scale only axis,
xmin=90,
xmax=180,
xlabel style={font=\color{white!15!black}},
xlabel={$\gamma{}_{\text{dB}}$},
ymin=0,
ymax=4,
ylabel style={font=\color{white!15!black}},
ylabel={Average rate   (bpcu)},
axis background/.style={fill=white},
title style={font=\bfseries},
title={},
xmajorgrids,
ymajorgrids,
legend style={nodes={scale=0.8},at={(0.03,0.97)}, anchor=north west, legend cell align=left, align=left, draw=white!15!black}
]
\addplot [color=mycolor1, line width=1.2pt]
  table[row sep=crcr]{%
90	0.0341418492547789\\
92	0.0453449198043537\\
94	0.0597184489325266\\
96	0.077883283200423\\
98	0.100694726603834\\
100	0.128247528723404\\
102	0.16188586897149\\
104	0.202532580753485\\
106	0.251304757421788\\
108	0.3100762997751\\
110	0.378955965282325\\
112	0.459119928285535\\
114	0.549893498385198\\
116	0.652911729067154\\
118	0.768736759201884\\
120	0.896964906819395\\
122	1.03630265178006\\
124	1.18415901087753\\
126	1.34156531085247\\
128	1.50564513363647\\
130	1.67411600160874\\
132	1.84353421146605\\
134	2.01252324977687\\
136	2.17744374966626\\
138	2.33538799987307\\
140	2.48470591821649\\
142	2.62247069837842\\
144	2.74598491154772\\
146	2.85497382506171\\
148	2.95033257544991\\
150	3.02956704834363\\
152	3.09490278791486\\
154	3.14787075366587\\
156	3.18797897558545\\
158	3.21872361752159\\
160	3.24145648269399\\
162	3.25868958455447\\
164	3.27090360250925\\
166	3.27948438969574\\
168	3.28544556892406\\
170	3.28954604306512\\
172	3.29251972954515\\
174	3.29502106482179\\
176	3.29615792007371\\
178	3.29642285037031\\
180	3.2967834768626\\
182	3.29768450797459\\
184	3.29720823290197\\
186	3.297455975493\\
188	3.29823072858678\\
190	3.29812549721018\\
};
\addlegendentry{QA-RZF}

\addplot [color=mycolor3, line width=1.2pt]
  table[row sep=crcr]{%
90	0.0342587717268177\\
92	0.0459642936332563\\
94	0.0606588304800265\\
96	0.0791237417186657\\
98	0.102126983543171\\
100	0.130013189106435\\
102	0.164207232247147\\
104	0.204623210249987\\
106	0.253610711865468\\
108	0.310977721607234\\
110	0.376861565206465\\
112	0.452140319546964\\
114	0.537792582626743\\
116	0.63256027728622\\
118	0.736658105389035\\
120	0.849106899573661\\
122	0.969158249482607\\
124	1.09474545522357\\
126	1.22312626719952\\
128	1.35000042665749\\
130	1.47346938480309\\
132	1.59065219497424\\
134	1.69440121719554\\
136	1.7835928729554\\
138	1.85391672943508\\
140	1.9039710324822\\
142	1.93514230312575\\
144	1.95074706557126\\
146	1.96027239418449\\
148	1.9734041862543\\
150	2.00181333341952\\
152	2.08770843031995\\
154	2.18962260911933\\
156	2.35975607611245\\
158	2.54453926813924\\
160	2.74063312210086\\
162	2.91756297569118\\
164	3.04886471553092\\
166	3.14165817441444\\
168	3.19245171804012\\
170	3.22595583786166\\
172	3.25057601733098\\
174	3.2727719601658\\
176	3.28716063301188\\
178	3.28730858378642\\
180	3.28849239351502\\
182	3.28840234780412\\
184	3.28847241810785\\
186	3.29111643753155\\
188	3.29074546129846\\
190	3.29111954606312\\
};
\addlegendentry{QI-RZF}

\addplot [color=mycolor4, line width=1.2pt]
  table[row sep=crcr]{%
90	0.0551889685464671\\
92	0.0731497111752915\\
94	0.0955628546277978\\
96	0.123283833801368\\
98	0.157452685585472\\
100	0.197753766501579\\
102	0.245735164463826\\
104	0.302262592404038\\
106	0.368698832727711\\
108	0.44552839031054\\
110	0.534799175171107\\
112	0.637716662964701\\
114	0.755082120692045\\
116	0.88843488846713\\
118	1.0382836527314\\
120	1.20595487313984\\
122	1.39057731714251\\
124	1.59133356810894\\
126	1.80649077335292\\
128	2.03359747439612\\
130	2.27061373068445\\
132	2.51260992677497\\
134	2.75300495423621\\
136	2.98596036617002\\
138	3.20365065951474\\
140	3.40047025542977\\
142	3.56939289636188\\
144	3.70646253587387\\
146	3.81116475991268\\
148	3.88595563815001\\
150	3.9347653125381\\
152	3.96446390742295\\
154	3.98147732468733\\
156	3.99029240055429\\
158	3.99497989180948\\
160	3.99728031686701\\
162	3.99858964553066\\
164	3.99931780613002\\
166	3.99958435573672\\
168	3.999759798531\\
170	3.99983988427458\\
};
\addlegendentry{RZF inf. res.}

\addplot [color=mycolor2, dashed, line width=1.2pt]
  table[row sep=crcr]{%
90	0.021549060890055\\
92	0.0284396289332202\\
94	0.0369876481961939\\
96	0.0475306033142686\\
98	0.0614089757992916\\
100	0.0784651500876793\\
102	0.0991521343500694\\
104	0.124697458011681\\
106	0.1547690437937\\
108	0.189759329032042\\
110	0.229870745831315\\
112	0.275622274803278\\
114	0.326796550946628\\
116	0.383572350968459\\
118	0.445911279476006\\
120	0.513540033614516\\
122	0.585808841244352\\
124	0.662006360022112\\
126	0.740989645963296\\
128	0.821696561179494\\
130	0.902748382528876\\
132	0.983584927730271\\
134	1.06175085864895\\
136	1.13698248768885\\
138	1.20810090994287\\
140	1.27453125713604\\
142	1.33560261229743\\
144	1.3905016332965\\
146	1.4390548343356\\
148	1.48194520578175\\
150	1.51921523816806\\
152	1.55082253503047\\
154	1.57839058442049\\
156	1.60121360471177\\
158	1.61981244835986\\
160	1.63443067234286\\
162	1.64552103405895\\
164	1.6542055825373\\
166	1.66068563270027\\
168	1.66518796590193\\
170	1.6686380269783\\
172	1.67118175564077\\
174	1.67293268810518\\
176	1.67405014844312\\
178	1.67456991254401\\
180	1.67503775250432\\
182	1.6755422013093\\
184	1.67577642745973\\
186	1.6758510012245\\
188	1.67577042729889\\
190	1.67593266359613\\
};

\addplot [color=mycolor3, dashed, line width=1.2pt]
  table[row sep=crcr]{%
90	0.0219281747991767\\
92	0.0289549662384033\\
94	0.0377328192858572\\
96	0.0491412402692412\\
98	0.0632034329122297\\
100	0.080342473493433\\
102	0.101541970063765\\
104	0.126231262046146\\
106	0.155394149292426\\
108	0.188837206791227\\
110	0.226828102880906\\
112	0.269829159935996\\
114	0.317386079971411\\
116	0.369415221723039\\
118	0.425718289845575\\
120	0.485563106829154\\
122	0.547762783407466\\
124	0.611864151362539\\
126	0.676234691494491\\
128	0.73998699293239\\
130	0.800409859546394\\
132	0.855813629563464\\
134	0.905508214906432\\
136	0.947781604452415\\
138	0.981101565160332\\
140	1.0043070422938\\
142	1.01714097014203\\
144	1.02015456974013\\
146	1.01447906230718\\
148	1.00250374932762\\
150	0.9813854193234\\
152	0.95190500804391\\
154	0.915427997127021\\
156	0.876172898326147\\
158	0.839348941106642\\
160	0.807854228719632\\
162	0.787050213088209\\
164	0.796409241800764\\
166	0.874006726308076\\
168	1.00000257331555\\
170	1.15719117608093\\
172	1.28001236043895\\
174	1.35719103304391\\
176	1.41130614728331\\
178	1.4366450270051\\
180	1.44832253185419\\
182	1.45281739891897\\
184	1.45522826286546\\
186	1.45505239414125\\
188	1.45522788429779\\
190	1.45494689115991\\
};

\addplot [color=mycolor4, dashed, line width=1.2pt]
  table[row sep=crcr]{%
90	0.0343370245222867\\
92	0.044685293540849\\
94	0.0574829508829734\\
96	0.0731322285305782\\
98	0.0921631813890758\\
100	0.115303648683831\\
102	0.14295187465902\\
104	0.176024244509193\\
106	0.215573164760172\\
108	0.262277771513855\\
110	0.316997921073645\\
112	0.380527732284877\\
114	0.453863295771166\\
116	0.538689715797331\\
118	0.634912607023399\\
120	0.744212988030708\\
122	0.866723101318504\\
124	1.00312986744162\\
126	1.15430556622298\\
128	1.31990843000335\\
130	1.49843566550812\\
132	1.68908793327421\\
134	1.89140107720239\\
136	2.10353329364587\\
138	2.32215742555901\\
140	2.54437914718735\\
142	2.76674826671189\\
144	2.9843894692167\\
146	3.19283952452629\\
148	3.38753946025374\\
150	3.56011040185669\\
152	3.70516506989238\\
154	3.81793312347534\\
156	3.89776425366079\\
158	3.94864755117191\\
160	3.97676661087364\\
162	3.99032201081767\\
164	3.99616197229522\\
166	3.99836933155171\\
168	3.99898894325047\\
170	3.99908948628845\\
172	3.99927670454447\\
174	3.99944355708734\\
176	3.99958705085187\\
178	3.99966482800587\\
180	3.99985079234049\\
182	3.99993346118618\\
184	3.99994133451005\\
186	3.99997837576693\\
188	3.99998775832101\\
190	3.99999843522431\\
};

\end{axis}
\end{tikzpicture}%
  \caption{Average rates of users in  MIMO downlink $N=64$, 16 QAM inputs and  $Q=4$. $M=8$ for solid curves, $M=32$ for dashed curves.}
   \label{fig:QAvsQI_Q}
\end{figure}

In this subsection, our proposed QA-RZF precoding is essentially compared to the quantization-ignorant RZF (QI-RZF) and a benchmark with no quantization. To obtain the QI-RZF curve, \Cref{alg:alter} is run with the assumption of no quantization -which can be done by setting $\xi_Q=1$ when computing $c_k$ values in \eqref{eq:c_k}- and then the precoded  signal goes through $Q$ level CE quantization.   
The benchmark  'RZF inf. res.' is  obtained by scaling the precoded signal of QI-RZF to have  power of $\gamma$ and then transmitting it with no quantization. Note that for QI-RZF, the approximate channel gains in \eqref{eq:lower_bound_m_2} grow very large at $\gamma\to\infty$. This leads users to almost achieve the natural rate limit, i.e., the maximum rate due to having finite constellation,  even with small values of $\cvec{w}$, for  which $\sum_{k=1}^K w_k<K$.  As a result, obtaining the exact solution of \eqref{eq:solution_1} requires an idealistic MMSE look-up table with infinite range. In such cases, the mercury/waterfilling implementation does not converge to satisfy  \eqref{eq:solution_1} and we manually set its solution to $\cvec{w}=\cvec{1}_M$. Otherwise, QI-RZF performs  even worse at $\gamma\to\infty$ in CE systems.

Fig.~\ref{fig:QAvsQI_Q}  illustrates  average rates of users in CE MIMO downlink with 16 QAM inputs, $N=64$ and $Q=4$. Solid curves depict the average rates for the system with $M=8$ users and  dashed curves are for $M=32$. The benefit of taking quantization into account is clearly observed when QA-RZF is compared to QI-RZF. For  $M=8$ and $M=32$, there are  gaps of approximately 5 dB and 6.8 dB around 1.8 and 1 bpcu, respectively. QI-RZF curve diverges from QA-RZF curve as transmit power increases. 

On the right side of Fig.~\ref{fig:QAvsQI_Q}, we see that the gap between QA-RZF and QI-RZF enlarges  when $M$ is increased  to $32$ from $8$, which can be interpreted  via the high transmit power analysis. Both QI-RZF and QA-RZF set $\cvec{w}=\cvec{1}_M$ at $\gamma\to\infty$ so that the gap is not due to the power allocation. On the other hand, from the discussion in \Cref{sec:high_tr} we know that the difference between $u$ values obtained for quantized and unquantized systems increases with number of users, which causes the gap between QA-RZF and QI-RZF to enlarge. 

Also, we observe that rates achieved by QI-RZF do not consistently increase with the transmit power in Fig.~\ref{fig:QAvsQI_Q}. This is caused  by the power allocation in QI-RZF at a particular region of high transmit power regime, where the mercury/waterfilling algorithm converges and $\cvec{w}=\cvec{1}_M$ is not yet set manually. As the approximate channel gains in \eqref{eq:lower_bound_m_2} grow large in  QI-RZF, the mercury/waterfilling converges to the power allocation that is inversely proportional to the channel gains \cite{LoTuVe05}. This power allocation is clearly far from the optimal uniform-like power allocation  of QA-RZF at high transmit power and causes the observed inconsistency.

\subsection{Performances with Different CE Quantization and Modulation Levels}
\begin{figure}[!]
\centering
\input{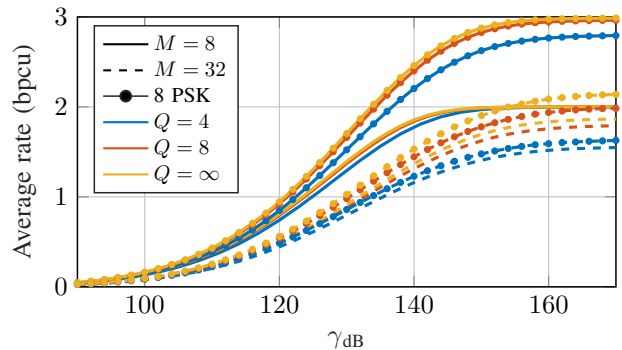}
\caption{Average rates of users in  CE quantized MIMO downlink with $N=64$. $M=8$ for solid curves, $M=32$ for dashed curves. Different colors represent different input  quantization levels. The curves with circle marks are for 8 PSK inputs.}
\label{fig:QPSK_8PSK_different_Q} 
\end{figure}
Fig.~\ref{fig:QPSK_8PSK_different_Q} depicts how QA-RZF performs with various quantization levels and input constellations.  The legend has a similar structure as the legend of Fig.~\ref{fig:bbvsbisect} with a difference that here  colors represent only the quantization levels and curves with circle marks are for 8 PSK inputs. Setting $Q=\infty$ corresponds to infinite resolution phase quantization with unit magnitude.  In Fig.~\ref{fig:QPSK_8PSK_different_Q} shift in $Q$ from $4$ to $8$ improves rates substantially, whereas increasing $Q$ above $8$ is not as effective especially at low user load or low/moderate transmit power. By comparing the cases with QPSK and 8 PSK inputs, we can see that increasing $Q$ makes a bigger difference in higher order modulation. 

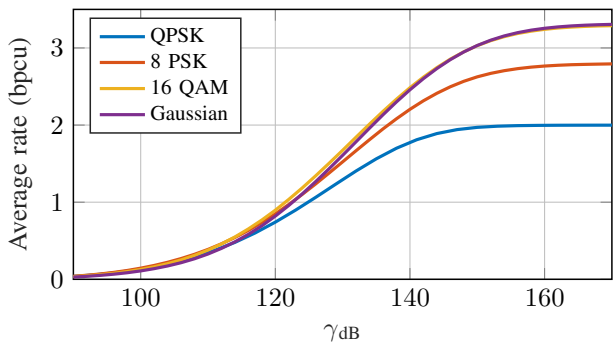
\begin{figure}[!]
\centering
%
%
\definecolor{mycolor1}{rgb}{0.00000,0.44700,0.74100}%
\definecolor{mycolor2}{rgb}{0.85098,0.32549,0.09804}%
\definecolor{mycolor3}{rgb}{0.92941,0.69412,0.12549}%
\definecolor{mycolor4}{rgb}{0.49412,0.18431,0.55686}%
\begin{tikzpicture}

\begin{axis}[%
width=0.951\fwidth,
height=0.477\fwidth,
at={(0\fwidth,0\fwidth)},
scale only axis,
xmin=90,
xmax=170,
xlabel style={font=\color{white!15!black}},
xlabel={$\gamma{}_{\text{dB}}$},
ymin=0,
ymax=3.5,
ylabel style={font=\color{white!15!black}},
ylabel={Average rate   (bpcu)},
axis background/.style={fill=white},
title style={font=\bfseries},
title={ },
xmajorgrids,
ymajorgrids,
legend style={nodes={scale=0.8},at={(0.03,0.97)}, anchor=north west, legend cell align=left, align=left, draw=white!15!black}
]
\addplot [color=mycolor1, line width=1.2pt]
  table[row sep=crcr]{%
90	0.0375541866463875\\
93	0.0569108684969414\\
96	0.0833604286677001\\
99	0.11869559240041\\
102	0.163459719243026\\
105	0.220279390642684\\
108	0.290361466625351\\
111	0.37557463759399\\
114	0.478711327546716\\
117	0.600886578615187\\
120	0.741332705894847\\
123	0.89656119366291\\
126	1.06312355462274\\
129	1.23507552703599\\
132	1.40485736449403\\
135	1.5624036256291\\
138	1.69969330021651\\
141	1.80980751856545\\
144	1.88988733430498\\
147	1.94121690296553\\
150	1.97033784510199\\
153	1.98538102229245\\
156	1.99265408184971\\
159	1.99595203299986\\
162	1.99762524532606\\
165	1.99830194928155\\
168	1.99863524642572\\
171	1.9987357004629\\
};
\addlegendentry{QPSK}

\addplot [color=mycolor2, line width=1.2pt]
  table[row sep=crcr]{%
90	0.0393989586016129\\
92	0.0525457861571136\\
94	0.0690665377762273\\
96	0.0897702838777983\\
98	0.114341696637308\\
100	0.144451314057746\\
102	0.179532701956107\\
104	0.220775621776723\\
106	0.268947172260118\\
108	0.324310819283459\\
110	0.38770330371248\\
112	0.460394883822486\\
114	0.542264480356286\\
116	0.633977062545009\\
118	0.736375598129664\\
120	0.84851295033385\\
122	0.97066133365708\\
124	1.10098466549553\\
126	1.23721701637344\\
128	1.37895022977529\\
130	1.52323810280417\\
132	1.6680998937728\\
134	1.81118142251264\\
136	1.9506422396105\\
138	2.08183727749697\\
140	2.20416292385959\\
142	2.31490981909697\\
144	2.41229331948571\\
146	2.49672814363893\\
148	2.56773319344836\\
150	2.62544964990692\\
152	2.6704054503439\\
154	2.70529476046074\\
156	2.73178364409301\\
158	2.75121474231354\\
160	2.76550933011956\\
162	2.77549604362937\\
164	2.78273636645059\\
166	2.78707214125483\\
168	2.79044746797191\\
170	2.79340140645074\\
};
\addlegendentry{8 PSK}

\addplot [color=mycolor3, line width=1.2pt]
  table[row sep=crcr]{%
90	0.0340769624141908\\
92	0.0454217649605582\\
94	0.0597061998866776\\
96	0.0778957476192191\\
98	0.100542032958161\\
100	0.128003103501104\\
102	0.161647571839195\\
104	0.202980739763776\\
106	0.251200773474752\\
108	0.310104916264591\\
110	0.378651926428303\\
112	0.458549236379994\\
114	0.549376144071695\\
116	0.653450288095623\\
118	0.768839856212199\\
120	0.896257044610509\\
122	1.03623395187417\\
124	1.1849410497301\\
126	1.3423859652312\\
128	1.50585922559538\\
130	1.67301872978455\\
132	1.84366012669114\\
134	2.0129427671423\\
136	2.17794768909656\\
138	2.33582561555701\\
140	2.4846258554642\\
142	2.62135927236905\\
144	2.74539263608107\\
146	2.85543666828673\\
148	2.94977697756646\\
150	3.02984568117547\\
152	3.09462600255275\\
154	3.14624643233279\\
156	3.18766544641095\\
158	3.21790201681313\\
160	3.24162458360675\\
162	3.25902603362789\\
164	3.27092612144446\\
166	3.27947458976603\\
168	3.28561936672675\\
170	3.28876596702958\\
};
\addlegendentry{16 QAM}

\addplot [color=mycolor4, line width=1.2pt]
  table[row sep=crcr]{%
90	0.0290828819484786\\
92	0.0387294776829866\\
94	0.0509676232853658\\
96	0.0661733208338362\\
98	0.0851816567769871\\
100	0.108722340655583\\
102	0.138057634701909\\
104	0.172929512996502\\
106	0.215147687479383\\
108	0.266494984043814\\
110	0.32790402677449\\
112	0.400957054409073\\
114	0.485973549700186\\
116	0.583172228505265\\
118	0.693697081806959\\
120	0.819503148330961\\
122	0.957694159225846\\
124	1.10774913092797\\
126	1.26730016050518\\
128	1.43505865250305\\
130	1.60788931647645\\
132	1.78403883287862\\
134	1.96071564645464\\
136	2.1337704018657\\
138	2.30059632365369\\
140	2.45756690261408\\
142	2.60219240453559\\
144	2.73291540926379\\
146	2.84854517601821\\
148	2.94840462198561\\
150	3.03221697023661\\
152	3.1008903485539\\
154	3.15583555778641\\
156	3.19851797151686\\
158	3.23096937688546\\
160	3.25510056315151\\
162	3.27300266948496\\
164	3.28591717144142\\
166	3.2950305055834\\
168	3.30138447473735\\
170	3.30575586518593\\
};
\addlegendentry{Gaussian}

\end{axis}
\end{tikzpicture}%
\caption{Average rates of users in  MIMO downlink $M=8,$ $N=64$ and  $Q=4$}
\label{fig:Different_mod_Q_4_M_8_N_64} 
\end{figure}

Fig~\ref{fig:Different_mod_Q_4_M_8_N_64} illustrates rates achieved by QA-RZF in the system with  $M=8$, $N=64$ and $Q=4$. Different curves stand for the systems with different input signal constellations. In  Fig~\ref{fig:Different_mod_Q_4_M_8_N_64},  the combination of two factors determines the performance of QA-RZF. The first factor is the positive effect of higher modulation order such that power allocation is more efficient compared to the low order modulated systems, since stronger users are away from the natural rate limit for a broader range of transmit power. For that reason, there is a tendency to achieve higher rates with higher order modulation. On the other hand, QA-RZF is based on an approximation obtained via \eqref{eq:high_transmit_power}. As modulation order increases \eqref{eq:high_transmit_power} gets more inaccurate at low transmit power. This is why 16 QAM  curve is outperformed by 8 PSK curve and curve for Gaussian inputs is outperformed by  8 PSK and 16 QAM curves at low transmit power. 
\subsection{Comparison to the State-of-the-Art}
Fig.~\ref{fig:Q_8_vs_sota} exhibits performance of  QA-RZF in comparison to the state-of-the art  method Q-GPI-SEM originally from \cite{ChPaLe22}, extended here in \Cref{sec:ex_QA_GPI_SEM}. Number of quantization levels is selected as $Q=8$. QA-RZF outperforms Q-GPI-SEM especially in systems with QPSK inputs. In systems with 16 QAM inputs, performances of QA-RZF and Q-GPI-SEM are very close. The comparison for Gaussian input has been shown in \cite{AsXuNo23_2}, where again QA-RZF and Q-GPI-SEM curves almost overlap. Also note that   results with $Q=4$ are essentially the same as  Fig.~\ref{fig:Q_8_vs_sota} and it is depicted in \cite{AsXuNo23_2} as well. QA-RZF outperforms Q-GPI-SEM in the systems with lower order modulation, since it can be easily adapted for every input constellation, whereas  Q-GPI-SEM is designed for Gaussian inputs . Finally, even though QA-RZF cannot outperform Q-GPI-SEM in higher order modulation, it has significantly less computational complexity as it is shown in \Cref{sec:complexity}.
\begin{figure}[!]
%
%
\definecolor{mycolor1}{rgb}{0.00000,0.44706,0.74118}%
\definecolor{mycolor2}{rgb}{0.85098,0.32549,0.09804}%
\definecolor{mycolor3}{rgb}{0.85000,0.32500,0.09800}%
\begin{tikzpicture}

\begin{axis}[%
width=0.951\fwidth,
height=0.477\fwidth,
at={(0\fwidth,0\fwidth)},
scale only axis,
xmin=90,
xmax=170,
xlabel style={font=\color{white!15!black}},
xlabel={$\gamma{}_{\text{dB}}$},
ymin=0,
ymax=4,
ylabel style={font=\color{white!15!black}},
ylabel={Average rate  (bpcu)},
axis background/.style={fill=white},
title style={font=\bfseries},
title={ },
xmajorgrids,
ymajorgrids,
legend style={nodes={scale=0.8},at={(0.03,0.97)}, anchor=north west, legend cell align=left, align=left, draw=white!15!black}
]
\addplot [color=mycolor1, line width=1.2pt]
  table[row sep=crcr]{%
90	0.0428946660209873\\
92	0.0561800308081239\\
94	0.0722475769573708\\
96	0.0922728460646854\\
98	0.116643762638051\\
100	0.14493140454578\\
102	0.17874355937029\\
104	0.217780234534097\\
106	0.261311832466541\\
108	0.313191785784807\\
110	0.3729747193191\\
112	0.438924856107691\\
114	0.513963423055993\\
116	0.599953581084039\\
118	0.692466562729416\\
120	0.793873030866076\\
122	0.902911536614188\\
124	1.01519930670796\\
126	1.13299168295181\\
128	1.25028821335015\\
130	1.37103021383946\\
132	1.4826235441037\\
134	1.58912431726225\\
136	1.68608291964134\\
138	1.7701807751341\\
140	1.83949889158281\\
142	1.89430573890054\\
144	1.93301820728743\\
146	1.95978121106241\\
148	1.97755309993624\\
150	1.98726439726313\\
152	1.99307791474034\\
154	1.99614606472834\\
156	1.99782838167185\\
158	1.99895514790206\\
160	1.99942897079568\\
162	1.99967970333308\\
164	1.99980431830181\\
166	1.99988691876853\\
168	1.99993700712328\\
170	1.99994091457227\\
};
\addlegendentry{QA-RZF}

\addplot [color=mycolor1, dashed, line width=1.2pt]
  table[row sep=crcr,forget plot]{%
90	0.0263456663517533\\
92	0.0340943450645524\\
94	0.0434403854836912\\
96	0.0549049099193652\\
98	0.0682631395500594\\
100	0.0848129526181047\\
102	0.103984865336481\\
104	0.12691634184389\\
106	0.153873137054677\\
108	0.185245721656881\\
110	0.221193614413959\\
112	0.262501494472191\\
114	0.308834253848905\\
116	0.361474597135349\\
118	0.420275354911869\\
120	0.483905111128147\\
122	0.555424974197664\\
124	0.630450322492834\\
126	0.712347575840124\\
128	0.796755293398629\\
130	0.883600190627517\\
132	0.972338245918647\\
134	1.06132513591478\\
136	1.14877388559765\\
138	1.23405162805395\\
140	1.31381049037493\\
142	1.38965666068964\\
144	1.45752548681793\\
146	1.52024925852909\\
148	1.5756197609514\\
150	1.62251946457896\\
152	1.66364012294513\\
154	1.69606196677442\\
156	1.72189367794865\\
158	1.74224334368078\\
160	1.75718608527967\\
162	1.76882811174493\\
164	1.77662216394582\\
166	1.78298693470112\\
168	1.78764432625468\\
170	1.79010463807316\\
};

\addplot [color=mycolor1, line width=1.2pt, mark=*, mark options={solid, fill=mycolor1, mycolor1}, mark size=0.7pt, forget plot]
  table[row sep=crcr]{%
90	0.0377357000730222\\
92	0.0499904028872956\\
94	0.0660850036873169\\
96	0.0857331318760863\\
98	0.110364356082031\\
100	0.14012780158007\\
102	0.176750984109636\\
104	0.220842119656347\\
106	0.274447770748191\\
108	0.338290387407312\\
110	0.412829165702559\\
112	0.499730469266276\\
114	0.599239221237582\\
116	0.712147339132502\\
118	0.839521564265389\\
120	0.979876112202369\\
122	1.13342800370104\\
124	1.29815897779453\\
126	1.47410997098991\\
128	1.65758404029964\\
130	1.84702716330691\\
132	2.03943039197272\\
134	2.23197030137576\\
136	2.42117506889698\\
138	2.60243666121225\\
140	2.77264640310996\\
142	2.9307537752157\\
144	3.07248610615308\\
146	3.19858975548068\\
148	3.30546554058169\\
150	3.39537586365356\\
152	3.46714872193756\\
154	3.52493540060229\\
156	3.56823323430137\\
158	3.60189863216622\\
160	3.62572276078382\\
162	3.64439026106203\\
164	3.6566214385058\\
166	3.66507925179743\\
168	3.67114237328095\\
170	3.67498325080392\\
};

\addplot [color=mycolor1, dashed, line width=1.2pt,mark=*, mark options={solid, fill=mycolor1, mycolor1}, mark size=0.7pt,forget plot]
  table[row sep=crcr]{%
90	0.0238458577790261\\
92	0.0313805015524117\\
94	0.0406754951192751\\
96	0.0524672404220771\\
98	0.0676691322200532\\
100	0.0859803910984736\\
102	0.109340070617319\\
104	0.137424114403775\\
106	0.170445031360753\\
108	0.20921682907091\\
110	0.254045673646134\\
112	0.305132330371545\\
114	0.36303687277141\\
116	0.427742181112181\\
118	0.500101200086482\\
120	0.578769220720649\\
122	0.66412562867278\\
124	0.755320384767532\\
126	0.851583255832453\\
128	0.951018452284452\\
130	1.05324691910647\\
132	1.15538098949042\\
134	1.25609477112267\\
136	1.35379021469058\\
138	1.44826360913169\\
140	1.53706207643671\\
142	1.62046932457002\\
144	1.69564192267628\\
146	1.76350030480654\\
148	1.82380806363545\\
150	1.87759479392358\\
152	1.92365380198306\\
154	1.96428679484687\\
156	1.99774454005314\\
158	2.02477173902489\\
160	2.04632411052123\\
162	2.06318509083751\\
164	2.07628510734422\\
166	2.08571436281951\\
168	2.09357992756049\\
170	2.09819746276706\\
172	2.10188889226634\\
174	2.10427279350739\\
176	2.10568588474349\\
178	2.10742857677009\\
180	2.10784232124822\\
182	2.10930460883767\\
184	2.10877022203851\\
186	2.10873365814654\\
188	2.10942435001401\\
190	2.10954882271341\\
};

\addplot [color=mycolor2, line width=1.2pt]
  table[row sep=crcr]{%
90	0.039349849650696\\
92	0.0514147063961492\\
94	0.0660768000474997\\
96	0.0840707654763682\\
98	0.10489648090901\\
100	0.129344129686317\\
102	0.15802065498441\\
104	0.19168089523078\\
106	0.229830402598524\\
108	0.27288134256283\\
110	0.321498228189827\\
112	0.376748599662868\\
114	0.440467791136817\\
116	0.511547966675297\\
118	0.590438657239656\\
120	0.676640568911553\\
122	0.770281544224609\\
124	0.870115995479263\\
126	0.973157638844828\\
128	1.07738978214115\\
130	1.1817193545769\\
132	1.28560281512101\\
134	1.38810973926341\\
136	1.48702492789218\\
138	1.57954423671646\\
140	1.66391173737568\\
142	1.73868847530656\\
144	1.80240161918728\\
146	1.85566362258762\\
148	1.89812338653189\\
150	1.93053059550318\\
152	1.9539554760193\\
154	1.97038303121882\\
156	1.98126112337704\\
158	1.98803927056059\\
160	1.99244053508186\\
162	1.9951941621808\\
164	1.99715168224291\\
166	1.99850974241902\\
168	1.99925908018313\\
170	1.99955244460692\\
};
\addlegendentry{QA-GPI-SEM}

\addplot [color=mycolor3, line width=1.2pt,mark=*, mark options={solid, fill=mycolor3, mycolor3}, mark size=0.7pt, forget plot]
  table[row sep=crcr]{%
90	0.0367989138139898\\
92	0.0487358757960205\\
94	0.0637405698008436\\
96	0.0824956148863495\\
98	0.105686918233012\\
100	0.133611901637923\\
102	0.168172700639995\\
104	0.210553879743656\\
106	0.260249852959062\\
108	0.318658256306944\\
110	0.388562466455473\\
112	0.470588165325421\\
114	0.567504525086435\\
116	0.676881240292472\\
118	0.799944901562595\\
120	0.939337906955031\\
122	1.09133197136369\\
124	1.25605914473644\\
126	1.43087472259868\\
128	1.61258189244287\\
130	1.80155556231248\\
132	1.9916507065424\\
134	2.18468425388448\\
136	2.37440019653349\\
138	2.55704902978802\\
140	2.73011245930736\\
142	2.8897502870156\\
144	3.03501691298937\\
146	3.16283206066347\\
148	3.2728878714456\\
150	3.36566505644041\\
152	3.44155651803106\\
154	3.50077556375293\\
156	3.54644601058998\\
158	3.58154523858128\\
160	3.60661944521\\
162	3.62447096177226\\
164	3.63613284868039\\
166	3.64538546854621\\
168	3.6516463479319\\
170	3.65572897953311\\
};

\addplot [color=mycolor3, dashed, line width=1.2pt,forget plot]
  table[row sep=crcr]{%
90	0.0235785329323547\\
92	0.0304066370927358\\
94	0.038737795830247\\
96	0.0483507932829558\\
98	0.059805777027224\\
100	0.072918861090759\\
102	0.0882725145329721\\
104	0.105852096517115\\
106	0.126641622993371\\
108	0.150426856292181\\
110	0.178744903069348\\
112	0.212473015800004\\
114	0.251872091218989\\
116	0.295736967079757\\
118	0.345803518244234\\
120	0.40097268035608\\
122	0.461518414618011\\
124	0.528356073411542\\
126	0.602226442054312\\
128	0.680706465870259\\
130	0.762957382257031\\
132	0.8483966637902\\
134	0.934583787277327\\
136	1.02052338410261\\
138	1.106862361829\\
140	1.19006213604258\\
142	1.26947338852848\\
144	1.3424772385677\\
146	1.40918880410072\\
148	1.46993013607409\\
150	1.52390465844748\\
152	1.57164584643297\\
154	1.6116987344045\\
156	1.64510924000588\\
158	1.67215624202619\\
160	1.6927850536686\\
162	1.70951404306122\\
164	1.72109053391941\\
166	1.73024278525192\\
168	1.73619084589282\\
170	1.74146840952795\\
};

\addplot [color=mycolor2, dashed, line width=1.2pt,mark=*, mark options={solid, fill=mycolor2, mycolor2}, mark size=0.7pt, forget plot]
  table[row sep=crcr]{%
90	0.0238506611967212\\
92	0.0311861068241859\\
94	0.0402966996881365\\
96	0.0516478941029033\\
98	0.0654825006188937\\
100	0.082405377591019\\
102	0.10231638918086\\
104	0.126255310612575\\
106	0.1548472773607\\
108	0.189662033614269\\
110	0.230764763985799\\
112	0.280798857594751\\
114	0.338714635038741\\
116	0.404966858078128\\
118	0.479293195672699\\
120	0.560064237291638\\
122	0.647309753415911\\
124	0.74005978302826\\
126	0.838386609579859\\
128	0.940323963364073\\
130	1.04415813504933\\
132	1.14901521215815\\
134	1.25145248956517\\
136	1.35163109292042\\
138	1.44657045325961\\
140	1.53530564451199\\
142	1.61846371240308\\
144	1.69264423582355\\
146	1.75903608781583\\
148	1.81806733261125\\
150	1.86860110103321\\
152	1.9129476023169\\
154	1.95028768790441\\
156	1.98046089312384\\
158	2.00570990943606\\
160	2.02529880582561\\
162	2.0399316263185\\
164	2.05233366341628\\
166	2.06104730817151\\
168	2.06775780762693\\
170	2.07238006312066\\
};
\end{axis}
\end{tikzpicture}%
   \centering

\caption{Average rates of users in the MIMO downlink with QA-RZF and Q-GPI-SEM precoding, $N=64$ and $Q=8$. $M=8$ for solid curves, $M=32$ for dashed curves. Curves with no markers are for  QPSK inputs, curves with circle markers are for 16 QAM.}
\label{fig:Q_8_vs_sota}
  
\end{figure}
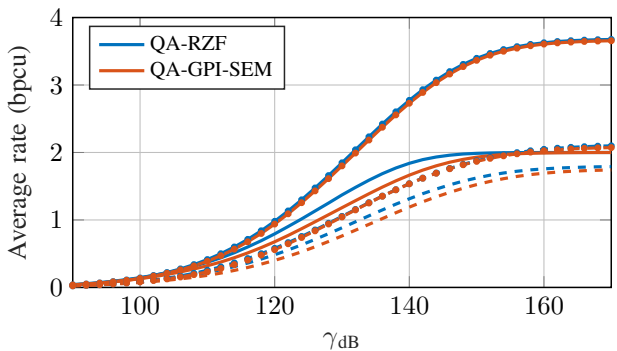

\section{Conclusion and Outlook}\label{sec:concl}
In this paper, we presented two algorithms (based on the branch and bound method and alternating optimization) for sum rate maximization  in the \ac{CE} \ac{MIMO} with  \ac{RZF} precoding. Although  alternating optimization obtains a local maximum for the approximate sum rate, it effectively achieves the same sum rate as the branch-and-bound method, which obtains the global maximum for the approximate sum rate. This makes us select the alternating optimization as the main proposed algorithm, since it has a lower computational complexity. To compare our algorithm to the state-of-the art, we extended the Q-GPI-SEM algorithm to the CE MIMO with higher resolution. The advantages of the proposed method over Q-GPI-SEM is its compatibility with any input signal distribution and significantly lower computational complexity.  In future studies, the proposed method can provide  a refined initial point for symbol-wise precoding techniques to account for the power allocation. Similarly, it can be used to improve the power allocation for hybrid precoding in CE MIMO downlink significantly.
\vspace{-2mm}
\begin{appendices} 
\section{Upper Bound to $R(\mathcal{S})$}\label{sec:upp_bound}
Let us  rewrite \eqref{eq:sum_rate_box}  as 
\begin{equation}\label{eq:ineq1}
R(\mathcal{S})=\underset{\cvec{w}\geq\cvec{0}}{\max} \sum_{k=1}^{K_{\mathcal{S}}} I(\lambda_k(u_\mathcal{S},K_\mathcal{S})w_k) \;\text{s.t.}\; \sum_{k=1}^{K_{\mathcal{S}}}w_k=K_{\mathcal{S}},
\end{equation}
where we denote the optimal $(u,K)$ for \eqref{eq:sum_rate_box}  as $(u_\mathcal{S},K_{\mathcal{S}})$ and $\cvec{w}^\star  (u_\mathcal{S},K_{\mathcal{S}})$ is the corresponding optimal power allocation obtained by \eqref{eq:solution_1}.
We define an upper bound $U(\mathcal{S})$ by using the following intermediate bounds 
\begin{align}\label{eq:ineq2}
\hat{R}(\mathcal{S})= \underset{\substack{\cvec{w}\geq \cvec{0}}}{\max} \sum_{k=1}^{K_\mathbb{S}} I(\Lambda_k w_k) \; \text{s.t.} \; \sum_{k=1}^{K_\mathcal{S}}w_k=K_{\mathcal{S}},\\
\label{eq:ineq3}
\hat{U}(\mathcal{S})= \underset{\substack{\cvec{w}\geq \cvec{0}}}{\max} \sum_{k=1}^{K_\mathbb{S}} I(\Lambda_k w_k) \; \text{s.t.} \; \sum_{k=1}^{K_\mathcal{S}}w_k=K_{\text{U}},\\
\label{eq:upper_bound_box_app}
U(\mathcal{S})=\underset{\substack{\cvec{w}\geq \cvec{0}}}{\max} \sum_{k=1}^{K_\text{U}} I(\Lambda_k w_k) \; \text{s.t.} \; \sum_{k=1}^{K_\text{U}}w_k=K_\text{U},
\end{align}
where we denote  $\Lambda_k(u_\text{L},u_\text{U},K_\text{L})$ as $\Lambda_k$ and
it holds that $U(\mathcal{S})\geq \hat{U}(\mathcal{S})\geq \hat{R}(\mathcal{S})\geq R(\mathcal{S})$. Note that we denote Inequality $\hat{R}(\mathcal{S})\geq R(\mathcal{S})$ is concluded by observing that \eqref{eq:ineq2} is a power allocation problem with same transmit power constraint as \eqref{eq:ineq1}, but with stronger channels due to \eqref{eq:strong_channel}. Inequality $\hat{U}(\mathcal{S})\geq \hat{R}(\mathcal{S})$ holds, since \eqref{eq:ineq3} is a power allocation problem for the same channels as \eqref{eq:ineq2} but with a higher power constraint. Finally, one can take the optimal power allocation vector of \eqref{eq:ineq3} and pad $K_\text{U}-K_\mathcal{S}$ zeros to make it a feasible power allocation vector for \eqref{eq:upper_bound_box_app}. The rate achieved with this feasible vector in  \eqref{eq:upper_bound_box_app} is equal to the optimal rate in \eqref{eq:ineq3} so that we  conclude  $U(\mathcal{S})\geq \hat{U}(\mathcal{S})$.

\section{Entries of $\cvec{w}_{\text{U}(\mathcal{S})}$ and the Global Maximum  }\label{sec:app}
\begin{theorem}\label{theorem:U_S} For a given subset $\mathcal{S}$, if  the optimal power allocation vector for \eqref{eq:upper_bound_box_app}, denoted as  $\cvec{w}_{\text{U}(\mathcal{S})}$, has less nonzero entries than $K_\text{L}$, then $\mathcal{S}$ does not contain the solution of \eqref{eq:prob_again}. We arrive to \Cref{theorem:U_S} by combining Lemmas 3 and 4.
\end{theorem}

\begin{lemma} Within a given subset $\mathcal{S}$, number of nonzero elements in the optimal power allocation vector $\cvec{w}^\star  (u_\mathcal{S},K_{\mathcal{S}})$ is upper  bounded by the number of nonzero elements in $\cvec{w}_{\text{U}}(\mathcal{S})$.
\end{lemma}
\begin{proof}
Let us denote number of nonzero entries in optimal $\cvec{w}$ vectors  as $J,\hat{J},\hat{J}_\text{U},J_\text{U}$ for  \eqref{eq:ineq1}, \eqref{eq:ineq2}, \eqref{eq:ineq3} and \eqref{eq:upper_bound_box_app}, respectively. From the argumentation in Appendix~\ref{sec:upp_bound}, it is clear that $J\leq\hat{J}\leq\hat{J}_\text{U}$. In order to compare $\hat{J}_\text{U}$ to $J_\text{U}$, let us consider the optimality conditions for  \eqref{eq:ineq3} and  \eqref{eq:upper_bound_box_app} respectively as follows
\begin{equation}\label{eq:opt_cond}
\sum_{\substack{k=1}}^{\hat{J}_{\text{U}}} \frac{1}{\Lambda_k(u_\text{L},u_\text{U},K_\text{L})} \text{MMSE}^{-1}(\frac{\hat{\mu}_\text{U}}{\Lambda_k(u_\text{L},u_\text{U},K_\text{L})})=K_\text{U},
\end{equation}
\begin{equation}\label{eq:opt_cond_2}
\sum_{\substack{k=1}}^{J_{\text{U}}} \frac{1}{\Lambda_k(u_\text{L},u_\text{U},K_\text{L})} \text{MMSE}^{-1}(\frac{{\mu}_\text{U}}{\Lambda_k(u_\text{L},u_\text{U},K_\text{L})})=K_\text{U}.
\end{equation}
Terms $\hat{\mu}_\text{U}$ and $\mu_\text{U}$ are the optimal waterlevels for \eqref{eq:ineq3} and  \eqref{eq:upper_bound_box_app}. Before proceeding recall that $\Lambda_1(u_\text{L},u_\text{U},K_\text{L})\geq\Lambda_2(u_\text{L},u_\text{U},K_\text{L})\ldots\geq\Lambda_{K_\text{U}}(u_\text{L},u_\text{U},K_\text{L})$. Now let us assume that $J_\text{U}<\hat{J}_{\text{U}}$. In this case, \eqref{eq:opt_cond} and \eqref{eq:opt_cond_2} can only hold if $\mu_\text{U}<\hat{\mu}_\text{U}$.  On the other hand, if the optimal power allocation vector in \eqref{eq:upper_bound_box_app} has  ${J}_{\text{U}}$ nonzero entries, then it holds that $\Lambda_{\hat{J}_\text{U}}(u_\text{L},u_\text{U},K_\text{L})\leq\mu_\text{U}<\Lambda_{{J}_\text{U}}(u_\text{L},u_\text{U},K_\text{L})$. Similarly, if the optimal power allocation vector in \eqref{eq:ineq2} has $\hat{J}_\text{U}$ nonzero entries, then $\hat{\mu}_{\text{U}}<\Lambda_{\hat{J}_\text{U}}(u_\text{L},u_\text{U},K_\text{L})$. As a result, it must also hold that $\hat{\mu}_\text{U}<\mu_\text{U}$, which leads to a contradiction such that $J_\text{U}<\hat{J}_{\text{U}}$ cannot hold. We conclude that $J\leq\hat{J}\leq\hat{J}_\text{U}\leq J_\text{U}$.
\end{proof}

\begin{lemma}\label{lemma:obs_2}For a subset $\mathcal{S}=\{(u,K)| \,u\in(u_\text{L},u_\text{U}),\,K\in\{K_{\text{L}},\ldots K_\text{U}\}\}$,  if sum rate maximizing power allocation vector  $\cvec{w}^{\star}(u_\mathcal{S},K_\mathcal{S})$ has less nonzero elements than $K_\text{L}$, then $\mathcal{S}$ does not contain the solution of \eqref{eq:prob_again}.
\end{lemma}
\begin{proof}
Let us assume that  $\cvec{w}^{\star}(u_\mathcal{S},K_\mathcal{S})$ has $J<K_\text{L}$ nonzero elements, which means that last $K_\mathcal{S}-J$ entries of  $\cvec{w}^{\star}(u_\mathcal{S},K_\mathcal{S})$ are equal to 0. For  subset $\mathcal{S}$, the optimal \ac{SNR} values of  users  reads as 
\begin{equation}\label{eq:SNR_Ks}
\text{SNR}_{k}=
\frac{\frac{N}{K_\mathcal{S}}-u^2_{\mathcal{S}}}{\frac{K_\mathcal{S}-1}{K_\mathcal{S}}(1-u_\mathcal{S})^2+c_k}w_k^{\star}(u_\mathcal{S},K_\mathcal{S}) 
\end{equation} 
for $k=1,2\ldots J$. For $k=J+1,\ldots K_\mathcal{S}$, $\text{SNR}_{k}=0$. Note that $\sum_{k=1}^{K_\mathcal{S}}w^{\star}_k(u_\mathcal{S},K_\mathcal{S})=K_\mathcal{S}$. Let us consider an alternative system with $u=u_\mathcal{S}$ and $K=J$, which is not in the subset $\mathcal{S}$. For such system, setting $\bar{w}_k=\frac{J}{K_\mathcal{S}}w^{\star}_k(u_\mathcal{S},K_\mathcal{S})$ for $k=1,\ldots J$, is a valid power  allocation as it satisfies $\sum_{k=1}^{J}\bar{w}_k=J$.  As a result, the \acp{SNR} values in system with $u_\mathcal{S}$, $J$ and $\bar{\cvec{w}}$ reads as
\begin{equation}\label{eq:SNR_J}
\bar{\text{SNR}}_{k}= \frac{\frac{N}{K_\mathcal{S}}-\frac{J}{K_\mathcal{S}}u^2_{\mathcal{S}}}{\frac{J-1}{J}(1-u_\mathcal{S})^2+c_k}w_k^{\star}(u_\mathcal{S},K_\mathcal{S}) 
\end{equation}
for $k=1,\ldots J$. It is clear that $\bar{\text{SNR}}_{k}>\text{SNR}_{k}$, which implies that the optimal rates achieved by active users in $\mathcal{S}$ are less than the rates achieved in the  alternative system with $u=u_\mathcal{S}$, $K=J$ and $\bar{\cvec{w}}$. As the alternative system is also included in the initial set $\mathcal{S}_0$, it is clear that $\mathcal{S}$ does not contain the solution of \eqref{eq:prob_again}.
\end{proof}
\end{appendices}
\bibliographystyle{IEEEtran}
\bibliography{IEEEabrv,examplebibfile} 

\end{document}